\documentclass[envcountsame,envcountsect,runningheads,11pt]{llncs}
\usepackage[a4paper,margin=3.40cm]{geometry}

\newcommand{\LongVersion}[1]{#1}
\newcommand{\ShortVersion}[1]{}

\usepackage{latexsym} 
\usepackage{xspace} 
\usepackage{amssymb} 
\usepackage{amsmath}
\usepackage{url} 
\input{xy} 
\xyoption{all} 
\usepackage[ruled,vlined,linesnumbered]{algorithm2e}
\usepackage{hyperref}
\usepackage{calc}

\newcommand{\myend}{\mbox{ }\hfill$\Box$}
\newcommand{\myEnd}{\mbox{ }\hfill$\Box$}
\def\ramka#1{\fbox{\parbox{\textwidth-0.8em}{#1}}}
\newenvironment{Definition}{\begin{definition}\em}{\end{definition}}

\newcommand{\V}{\forall}
\newcommand{\E}{\exists}
\newcommand{\Dmd}{\Diamond}
\newcommand{\lDmd}[1]{\langle#1\rangle}

\newcommand{\ovl}[1]{\overline{#1}}
\newcommand{\fracc}[2]{\displaystyle{\frac{\;#1\;}{\;#2\;}}}

\newcommand{\EXPTIME}{\textsc{Exp\-Time}\xspace}
\newcommand{\NEXPTIME}{\textsc{NExp\-Time}\xspace}

\newcommand{\mM}{\mathcal{M}} 
\newcommand{\mL}{\mathcal{L}}

\def\defeq{=}

\def\eqref#1{(\ref{#1})}

\newcommand{\comment}[1]{} 

\newcommand{\CPDLreg}{CPDL$_{reg}$\xspace}

\newcommand{\ALC}{$\mathcal{ALC}$\xspace} 
\newcommand{\ALCreg}{$\mathcal{ALC}_{reg}$\xspace}

\newcommand{\SHIO}{$\mathcal{SHIO}$\xspace}

\newcommand{\HPDL}{HPDL\xspace}
\newcommand{\CHPDL}{$\mathcal{C}_{\mathrm{HPDL}}$\xspace}

\newcommand{\rAnd}{(\land)}
\newcommand{\rOr}{(\lor)}
\newcommand{\rAutB}{(aut_\Box)}
\newcommand{\rAutD}{(aut_\Dmd)}
\newcommand{\rBox}{([\mathrm{A}])}
\newcommand{\rBoxD}{([\mathrm{A}]_f)}
\newcommand{\rDmd}{(\lDmd{\mathrm{A}})}
\newcommand{\rDmdD}{(\lDmd{\mathrm{A}}_f)}
\newcommand{\rBoxQm}{(\Box_{?})}
\newcommand{\rDmdQm}{(\Dmd_{?})}
\newcommand{\rTrans}{(\mathit{trans})}
\newcommand{\rUnsat}{(\mathit{close})}

\newcommand{\rFormingState}{(\mathit{form\textrm{-}state})}
\newcommand{\rBoxTrans}{(\Box_{\mathrm{trans}})}
\newcommand{\rReexpand}{(\mathit{re\textrm{-}expand})}
\newcommand{\rReplNom}{(\mathit{repl\textrm{-}nom})}
\newcommand{\rNom}{(\mathit{nominal})}

\newcommand{\props}{\mathcal{PROP}}
\newcommand{\mindices}{\Sigma}
\newcommand{\noms}{\mathcal{O}}

\newcommand{\bsf}{\mathit{bsf}}
\newcommand{\cls}{\mathit{closure}}
\newcommand{\clsZ}{\mathit{closure}_0}

\newcommand{\Aut}{\mathbb{A}}

\newcommand{\Type}{\mathit{Type}}
\newcommand{\Status}{\mathit{Status}}
\newcommand{\Label}{\mathit{Label}}
\newcommand{\FullLabel}{\mathit{FullLabel}}

\newcommand{\RFormulas}{\mathit{Reduced}}

\newcommand{\State}{\mathit{state}}
\newcommand{\NonState}{\mathit{non\textrm{-}state}}

\newcommand{\Unsat}{\mathit{closed}}
\newcommand{\Incomplete}{\mathit{incomplete}}
\newcommand{\Expanded}{\mathit{expanded}}
\newcommand{\Unexpanded}{\mathit{unexpanded}}
\newcommand{\Null}{\mathit{null}}

\newcommand{\Blocked}{\mathit{blocked}}
\newcommand{\UnsatWrt}{\mathit{closed\textrm{-}wrt}}
\newcommand{\UnsatWrtP}[1]{\mathit{closed\textrm{-}wrt}_{_+}(#1)}
\newcommand{\Repl}{\mathit{NomRepl}}
\newcommand{\SType}{\mathit{SType}}
\newcommand{\Complex}{\mathit{complex}}
\newcommand{\Simple}{\mathit{simple}}
\newcommand{\AssSN}{\mathit{AssSugByNom}}

\SetKwInput{Input}{Input}
\SetKwInput{Output}{Output}
\SetKwInput{GlobalData}{Global data}
\SetKwInput{Purpose}{Purpose}
\SetKw{Call}{call}
\SetKw{Set}{set}
\SetKw{Exit}{exit}
\SetKw{Break}{break}
\SetKw{Continue}{continue}

\SetKwFunction{NewSucc}{NewSucc} 
\SetKwFunction{FindProxy}{FindProxy} 
\SetKwFunction{ConToSucc}{ConToSucc} 
\SetKwFunction{Trans}{Trans}
\SetKwFunction{Apply}{Apply}
\SetKwFunction{Tableau}{Tableau}
\SetKwFunction{Expand}{Expand}
\SetKwFunction{UpdateStatus}{UpdateStatus}
\SetKwFunction{PropagateStatus}{PropagateStatus}
\SetKwFunction{Kind}{Kind}
\SetKwFunction{AfterTrans}{AfterTrans}
\SetKwFunction{BeforeFormingState}{BeforeFormingState}
\SetKwFunction{ToExpand}{ToExpand}
\SetKwFunction{AFormulas}{AFmls}
\SetKwFunction{NDFormulas}{NDFmls}
\SetKwFunction{Expand}{Expand}
\SetKwFunction{ExpandAndRealize}{ExpandAndRealize}
\SetKwFunction{Realize}{Realize}

\newcommand{\EdgeLabels}{\mathit{EdgeLabels}}
\newcommand{\ELabels}{\mathit{ELabels}}

\def\tuple#1{(#1)}

\newcommand{\TSDR}{\mathit{TimeStampDR}}

\newcommand{\realized}{\mathit{realized}}

\renewcommand{\subsubsection}[1]{
	\bigskip
	
	\noindent
	{\bf #1:}\\[1ex]	
}

\title{ExpTime Tableaux with Global Caching for Hybrid PDL} 

\titlerunning{\EXPTIME Tableaux with Global Caching for Hybrid PDL} 

\author{Linh Anh Nguyen} 
\institute{
	Institute of Informatics, University of Warsaw \\
	Banacha 2, 02-097 Warsaw, Poland\\
	\email{nguyen@mimuw.edu.pl}\\[2ex]
\today
}

\authorrunning{L.A. Nguyen}	 

\begin{document} 
\maketitle  
\sloppy 

\begin{abstract}
	We present the first direct tableau decision procedure with the ExpTime complexity for HPDL (Hybrid Propositional Dynamic Logic). It checks whether a given ABox (a finite set of assertions) in HPDL is satisfiable. Technically, it combines global caching with checking fulfillment of eventualities and dealing with nominals. Our procedure contains enough details for direct implementation and has been implemented for the TGC2 (Tableaux with Global Caching) system. As HPDL can be used as a description logic for representing and reasoning about terminological knowledge, our procedure is useful for practical applications. 
\end{abstract}


\section{Introduction}

Propositional dynamic logic (PDL) \cite{FischerLadner79,HKT00} is one of the most well-known modal logics. It was designed for reasoning about correctness of programs, but can be modified or extended for other purposes. For example, its extension \CPDLreg with converse and regular inclusion axioms can be used as a framework for multi-agent logics~\cite{DKNS11}. As another example, the description logic \ALCreg is a variant of PDL for representing and reasoning about terminological knowledge~\cite{Schild91}. Extensions of \ALCreg were also studied by researchers, including the ones with regular inclusion axioms, inverse roles, qualified number restrictions and nominals. 
%
%

Automated reasoning in PDL and its extensions is useful for practical applications. The first tableau decision procedure for PDL was developed by Pratt~\cite{Pratt80}. It implicitly uses global caching and has the \EXPTIME (optimal) complexity. Nguyen and Sza{\l}as~\cite{NguyenS10FI} reformulated that procedure by explicitly using global caching and extended it for dealing with checking satisfiability of an ABox (a finite set of assertions) in PDL. Abate et al.~\cite{AbateGW09} gave another tableau decision procedure with global caching for PDL, which updates fulfillment of eventualities (i.e., existential star modalities) on-the-fly. Due to global caching, the procedures given in~\cite{AbateGW09,NguyenS10FI} have the \EXPTIME complexity. 

Tableaux with global caching were formally formulated by Gor{\'{e}} and Nguyen for the description logic \ALC~\cite{GoreN13} and extended for other logics. A tableau with global caching for checking satisfiability of a concept w.r.t.\ a TBox in \ALC is a rooted ``and-or'' graph, where the label of each node is a set of concepts treated as requirements to be realized for the node and each edge departing from an ``and''-node is labeled by a set of roles.\footnote{A concept is like a formula and a role is like an atomic program in PDL.} Using global caching, each node has a unique label, which means that before creating a new node we check whether there already exists a node with the same label that can be used as a proxy. It is sufficient to deterministically construct one ``and-or'' graph and update the statuses of the nodes by detecting direct clashes and propagating them backward appropriately. 
Extending tableaux with global caching for PDL~\cite{AbateGW09,NguyenS10FI}, the ability to check fulfillment of eventualities over the constructed ``and-or'' graph is essential. Dealing with ABoxes in PDL~\cite{NguyenS10FI}, an ``and-or'' graph has two kinds of nodes: complex nodes and simple nodes. The label of a complex node is an ABox (with assertions about different states and accessibility between them), while the label of a simple node is a set of formulas (about one state). Edges departing from complex ``and''-nodes lead to simple nodes and are labeled with more information. 


HPDL (Hybrid PDL) is the modal logic that extends PDL with nominals. It is more expressive than PDL and belongs to the same complexity class \EXPTIME-complete as PDL~\cite{FischerLadner79,SattlerV01} (regarding the satisfiability problem). 
Kaminski and Smolka~\cite{KaminskiS14} developed a decision procedure for HPDL. The authors used the term ``goal-directed'' to refer to the property that the search is done analytically (i.e., closely based on the input). Like the traditional tableau method for description logics, the search space used in~\cite{KaminskiS14} is an ``or''-tree of nodes that are ``and''-structures, where the ``or''-tree is generated by nondeterministic choices (using backtracking) and the ``and''-structures are demo graphs (like model graphs or Hintikka structures without statically reduced assertions). Nodes of a demo graph, called ``normal clauses'' in~\cite{KaminskiS14}, are formula sets without statically reduced formulas. Caching is done only within an ``and''-structure. This is similar to the technique used in Donini and Massacci's tableau algorithm for the description logic \ALC~\cite{donini-massacci-exptime-tableau-for-alc} and somehow equivalent to the ``anywhere blocking'' technique used in the traditional tableau method for description logics~\cite{BaaderSattler01}. Without a surprise, the decision procedure given in~\cite{KaminskiS14} for HPDL has the \NEXPTIME (non-optimal) complexity for the worst case. Devising efficient \EXPTIME decision procedures for HPDL and its extensions is claimed in~\cite{KaminskiS14} as an open problem.

In this article, we present the first tableau decision procedure with the \EXPTIME complexity for checking whether a given ABox (a finite set of assertions) in HPDL is satisfiable. Technically, it combines global caching with checking fulfillment of eventualities and the technique of dealing with nominals from our work on tableaux for the description logic \SHIO~\cite{SHIO}. 
Global caching not only guarantees the \EXPTIME complexity, but is also an important optimization technique for increasing efficiency. Our decision procedure for HPDL also uses other advanced techniques, for example, automaton-modal operators. 

In the absence of nominals or the ability to express global assumptions, the problem of checking satisfiability of an ABox is usually more general than the problem of checking satisfiability of a formula. In HPDL, the former problem is reducible to the latter. We consider the former instead of the latter due to the nature of our tableau method (which uses both complex nodes and simple nodes for tableaux). 

Our decision procedure for HPDL contains enough details for direct implementation and has been implemented for TGC2~\cite{TGC2}, which is a system based on tableaux with global caching for automated reasoning in modal and description logics. This system was designed and implemented with several optimization techniques. Regarding memory management, experiments showed that the amount of memory used by TGC2 is competitive with the ones used by the other reasoners. 
As far as we know, TGC2 is the first implemented system that can be used for reasoning in HPDL. 
We refer the reader to~\cite{TGC2} for details of the design of this system. 

The rest of this article is structured as follows. In Section~\ref{section: prel}, we present the syntax and semantics of HPDL and recall automaton-modal operators~\cite{HKT00,DKNS11}. We omit the feature of ``global assumptions'' as they can be expressed in PDL (by ``local assumptions''). In Section~\ref{section: tableaux}, we present our tableau calculus for HPDL, starting with the data structure, the tableau rules and ending with the corresponding tableau decision procedure and its properties. 
In Section~\ref{sec: example}, we present an example to illustrate our procedure. We estimate the complexity of the procedure and prove its soundness and completeness in Section~\ref{section: Proofs}. Concluding remarks are given in Section~\ref{section: conc}.


\section{Preliminaries}\label{section: prel}

\subsection{Hybrid Propositional Dynamic Logic}
\label{section: defs HPDL}

We use $\mindices$ to denote the set of {\em atomic programs}, $\props$ to denote the set of {\em propositions} (i.e., {\em atomic formulas}), and $\noms$ to denote the set of {\em nominals}. We denote elements of $\mindices$ by letters like $\sigma$ and $\varrho$, elements of $\props$ by letters like $p$ and~$q$, and elements of $\noms$ by letters like $a$ and~$b$.

A {\em Kripke model} is a~pair $\mM = (\Delta^\mM,\cdot^\mM)$, where $\Delta^\mM$ is a~set of {\em states}
and $\cdot^\mM$ is an interpretation function that maps each nominal $a \in \noms$ to an element $a^\mM$ of $\Delta^\mM$, 
each proposition $p \in \props$ to a~subset $p^\mM$ of $\Delta^\mM$ and each atomic program $\sigma \in \mindices$ to a~binary relation $\sigma^\mM$ on $\Delta^\mM$. Intuitively, $p^\mM$ is the set of states in which $p$ is true and $\sigma^\mM$ is the binary relation consisting of pairs (input$\!\_$state, output$\!\_$state) of the program~$\sigma$. 
%

{\em Formulas} and {\em programs} of the {\em base language} of \HPDL are defined by the following grammar rules, respectively, where $p \in \props$, $a \in \noms$ and $\sigma \in \mindices$:
\[
\begin{array}{rcl}
\varphi & ::= &
\top
\mid \bot
\mid p
\mid a
\mid \lnot \varphi
\mid \varphi \land \varphi
\mid \varphi \lor \varphi
\mid \varphi \to \varphi 
\mid \lDmd{\alpha}\varphi
\mid [\alpha]\varphi\\[1.0ex]
\alpha & ::= &
\sigma
\mid \alpha;\alpha
\mid \alpha \cup \alpha
\mid \alpha^*
\mid \varphi?
\end{array}
\]

We use letters like $\alpha$, $\beta$ to denote programs and letters like $\varphi$, $\psi$, $\chi$ to denote formulas. 
The intended meaning of program operators is as follows:
\begin{itemize}
	\item $\alpha;\beta$ stands for the sequential composition of $\alpha$ and $\beta$,
	\item $\alpha \cup \beta$ stands for the non-deterministic choice between $\alpha$ and $\beta$, 
	\item $\alpha^*$ stands for the repetition of $\alpha$ a non-deterministic number of times, 
	\item $\varphi?$ stands for checking whether $\varphi$ holds for the current state. 
\end{itemize}

Informally, a formula $\lDmd{\alpha}\varphi$ represents the set
of states $x$ such that the program $\alpha$ has a transition
from $x$ to a state $y$ satisfying $\varphi$. Dually, a formula
$[\alpha]\varphi$ represents the set of states $x$ from which
every transition of $\alpha$ leads to a state
satisfying~$\varphi$.
A formula $a$ (a nominal) represents the set consisting of the only state specified by~$a$. 

Formally, the interpretation function of a Kripke model $\mM$
is extended to interpret complex formulas and complex programs
as shown in Figure~\ref{fig: int-comp}. 
%
%
For a~set $\Gamma$ of formulas, we denote $\Gamma^\mM = \bigcap\{\varphi^\mM \mid$ $\varphi \in \Gamma\}$. If $w \in \varphi^\mM$ (resp.\ $w \in \Gamma^\mM$), then we say that $\varphi$ (resp.\ $\Gamma$) is {\em satisfied at} $w$ in~$\mM$. 
If there exists a Kripke model $\mM$ such that $\varphi^\mM$ (resp.\ $\Gamma^\mM$) is not empty, then $\varphi$ (resp.\ $\Gamma$) is {\em satisfiable}. 

\begin{figure}[t]
	\ramka{ \[
		\begin{array}{rclrcl}
		(\alpha;\beta)^\mM & = & \alpha^\mM \circ \beta^\mM &
		(\alpha \cup \beta)^\mM & = & \alpha^\mM \cup \beta^\mM\\
		(\alpha^*)^\mM & = & (\alpha^\mM)^* & 
		(\varphi?)^\mM & = & \{ (x,x)\mid x \in \varphi^\mM \} \\[1.5ex]
		\top^\mM & = & \Delta^\mM\;\;\; \bot^\mM = \emptyset &
		a^\mM & = & \{a^\mM\}\\
		(\lnot\varphi)^\mM & = & \Delta^\mM \setminus \varphi^\mM &
		(\varphi \to \psi)^\mM & = & (\lnot\varphi \lor \psi)^\mM \\
		(\varphi \land \psi)^\mM & = & \varphi^\mM \cap \psi^\mM 
		\quad\quad\quad &
		(\varphi \lor \psi)^\mM & = & \varphi^\mM \cup \psi^\mM\\[1.5ex]
		(\lDmd{\alpha}\varphi)^\mM & = & 
		\multicolumn{4}{l}{ \{ x \in \Delta^\mM \mid
			\E y((x,y) \in \alpha^\mM \land y \in \varphi^\mM) \}
		} \\
		([\alpha]\varphi)^\mM & = & 
		\multicolumn{4}{l}{ \{ x \in \Delta^\mM \mid
			\V y((x,y) \in \alpha^\mM \to y \in \varphi^\mM) \}
		}
		\end{array}
		\] }
	\caption{Interpretation of complex programs and complex formulas.}
	\label{fig: int-comp}
\end{figure}

An {\em assertion} is an expression of the form $a\!:\!\varphi$ or $\sigma(a,b)$. 
An {\em ABox} is a finite set of assertions. 
Let $\Null\!:\!\varphi$ stand for $\varphi$. By letters like $o$, $o_1$, $o_2$ we will denote nominals or $\Null$, and by letters like $\xi$, $\zeta$ we will denote formulas or assertions. 

We define:
\[
\begin{array}{rclcl}
\mM & \models & a\!:\!\varphi & \;\;\textrm{iff}\;\; & a^\mM \in \varphi^\mM, \\
\mM & \models & \sigma(a,b) & \;\;\textrm{iff}\;\; & (a^\mM,b^\mM) \in \sigma^\mM.
\end{array}
\]
If $\mM \models \xi$, then we say that $\mM$ {\em satisfies} $\xi$. 
We say that $\mM$ {\em satisfies} and is a {\em model} of an ABox $\Gamma$, and $\Gamma$ is {\em satisfied} in $\mM$, denoted by $\mM \models \Gamma$, if $\mM$ satisfies all assertions in~$\Gamma$. If $\Gamma$ is satisfied in some Kripke model $\mM$, then it is {\em satisfiable}.

Formulas $\varphi$ and $\psi$ are {\em equivalent}, denoted by $\varphi \equiv \psi$, if $\varphi^\mM = \psi^\mM$ for every Kripke model $\mM$. Assertions $\xi$ and $\zeta$ are {\em equivalent}, denoted $\xi \equiv \zeta$, if for every Kripke model $\mM$, $\mM \models \xi$ iff $\mM \models \zeta$.  

A formula/assertion is in the {\em negation normal form} (NNF) if it does not use $\to$ and it uses $\lnot$ only immediately before propositions or nominals. Every formula/assertion can be translated in polynomial time to an equivalent formula/assertion in NNF. From now on, by $\overline{\varphi}$ we denote the NNF of~$\lnot\varphi$. For an assertion $\xi = a\!:\!\varphi$, by $\overline{\xi}$ we denote $a\!:\!\overline{\varphi}$. An ABox is in NNF if all of its assertions are in NNF. 


\subsection{Automaton-Modal Operators}

The {\em alphabet $\Sigma(\alpha)$} of a program $\alpha$ and the {\em regular language $\mL(\alpha)$} generated by $\alpha$ are specified as follows:\footnote{Note that $\Sigma(\alpha)$ contains not only atomic programs but also expressions of the form~$(\varphi?)$, and a program $\alpha$ is a regular expression over its alphabet $\Sigma(\alpha)$.}

\[
\begin{array}{l@{\extracolsep{3em}}l}
\Sigma(\sigma) = \{\sigma\} & \mL(\sigma) =\{\sigma\}\\
\Sigma(\varphi?) = \{\varphi?\} & \mL(\varphi?) = \{\varphi?\}\\
\Sigma(\beta;\gamma) = \Sigma(\beta) \cup       \Sigma(\gamma) & \mL(\beta;\gamma) =
\mL(\beta).\mL(\gamma)\\
\Sigma(\beta \cup \gamma) = \Sigma(\beta)   \cup \Sigma(\gamma) & \mL(\beta \cup
\gamma) = \mL(\beta) \cup \mL(\gamma)\\
\Sigma(\beta^*) = \Sigma(\beta) & \mL(\beta^*) = (\mL(\beta))^*
\end{array}
\]
where for sets $M$ and $N$ of words, $M.N \defeq \{\alpha\beta \mid \alpha \in M, \beta \in N\}$, $M^0 \defeq \{\varepsilon\}$ ($\varepsilon$ denotes the empty word), $M^{n+1} \defeq M.M^n$ for $n \geq 0$, and $M^* \defeq \bigcup_{n \geq 0} M^n$.

We will use letters like $\omega$ to denote either an atomic program from $\mindices$ or a~test (of the form $\varphi?$).
A~word $\omega_1\ldots\omega_k \in \mL(\alpha)$ can be treated as the program $(\omega_1;\ldots;\omega_k)$, especially when interpreted in a~Kripke model.

Recall that a {\em finite automaton} $A$ over an alphabet $\Sigma(\alpha)$ is a tuple $\langle \Sigma(\alpha), Q, I, \delta, F\rangle$, where $Q$ is a finite set of states, $I \subseteq Q$ is the set of initial states, $\delta \subseteq Q \times \Sigma(\alpha) \times Q$ is the transition relation, and $F \subseteq Q$ is the set of accepting states.
A {\em run} of $A$ on a word $\omega_1 \ldots \omega_k$ is a finite sequence of states $q_0, q_1, \ldots, q_k$ such that $q_0 \in I$ and $\delta(q_{i-1},\omega_i,q_i)$ holds for every $1 \leq i \leq k$. It is an {\em accepting run} if $q_k \in F$.
We say that $A$ {\em accepts} a word $w$ if there exists an accepting run of $A$ on $w$. The set of words accepted by $A$ is denoted by~$\mL(A)$.

We will use the following convention:
\begin{itemize}
	\item given a~finite automaton $A$, we always assume that $A = \tuple{\Sigma_A, Q_A, I_A, \delta_A, F_A}$, 
	\item for $q \in Q_A$, we define $\delta_A(q) = \{(\omega,q') \mid (q,\omega,q') \in \delta_A\}$.
\end{itemize}

As a~finite automaton $A$ over an alphabet $\Sigma(\alpha)$ corresponds to a~program (the regular expression recognizing the same language), it is interpreted in a~Kripke model $\mM$ as follows:
\begin{equation}
	A^\mM = \bigcup \{\gamma^\mM \mid \gamma \in \mL(A)\}.\label{eq:sem-aut}
\end{equation}

For each program $\alpha$, let $\Aut_\alpha$ be a~finite automaton recognizing the regular language $\mL(\alpha)$. The automaton $\Aut_\alpha$ can be constructed from $\alpha$ in polynomial time.
We extend the base language with the auxiliary modal operators $[A,q]$ and $\lDmd{A,q}$, where $A$ is $\Aut_\alpha$ for some program $\alpha$ and $q$ is a~state of $A$.
Here, $[A,q]$ and $\lDmd{A,q}$ stand respectively for $[(A,q)]$ and $\lDmd{(A,q)}$, where $(A,q)$ is the automaton that differs from $A$ only in that $q$ is its only initial state.
We call $[A,q]$ (resp.\ $\lDmd{A,q}$) a~{\em universal} (resp.\ {\em existential}) {\em automaton-modal operator}.

In the {\em extended language} of HPDL, if $\varphi$ is a~formula, then $[A,q]\varphi$ and $\lDmd{A,q}\varphi$ are also formulas. The semantics of these formulas are defined as usual, treating $(A,q)$ as a~program with semantics specified by~\eqref{eq:sem-aut}. 
From now on, the extended language is used instead of the base language. 

Given a~Kripke model $\mM$ and a~state $x \in \Delta^\mM$, we have $x \in ([A,q]\varphi)^\mM$ (resp.\  $x \in (\lDmd{A,q}\varphi)^\mM$) iff
\begin{quote}
	$x_k \in \varphi^\mM$ for all (resp.\ some) $x_k \in \Delta^\mM$ such that there exist a~word $\omega_1\ldots\omega_k$ (with $k \geq 0$) accepted by $(A,q)$ with $(x,x_k) \in (\omega_1;\ldots;\omega_k)^\mM$.
\end{quote}
The condition $(x,x_k) \in (\omega_1;\ldots;\omega_k)^\mM$ means there exist states $x_0 = x$, $x_1, \ldots, x_{k-1}$ of $\mM$ such that, for each $1 \leq i \leq k$, if $\omega_i \in \mindices$ then $(x_{i-1},x_i) \in \omega_i^\mM$, else $\omega_i = (\psi_i?)$ for some $\psi_i$ and $x_{i-1} = x_i$ and $x_i \in \psi_i^\mM$. Clearly, $\lDmd{A,q}$ is dual to $[A,q]$ in the sense that $\lDmd{A,q}\varphi \equiv \lnot[A,q]\lnot\varphi$ for any formula~$\varphi$.


\section{A Tableau Calculus for \HPDL}\label{section: tableaux}

From now on, let $\Gamma$ be an ABox in NNF. In this section, we present a tableau calculus \CHPDL for checking whether $\Gamma$ is satisfiable. We specify the data structure, the tableau rules, the corresponding tableau decision procedure and state its properties. 

\subsection{The Data Structure}

Let $\EdgeLabels$ be the set of formulas and assertions of the form $\lDmd{\sigma}\varphi$ or $a\!:\!\lDmd{\sigma}\varphi$. 

\begin{Definition}
A {\em tableau} is a rooted graph $G = \tuple{V,E,\nu}$, where $V$ is a set of nodes, \mbox{$E \subseteq V \times V$} is a set of edges, $\nu \in V$ is the root, each node $v \in V$ has a number of attributes, and each edge $\tuple{v,w}$ may be labeled by a set $\ELabels(v,w) \subseteq \EdgeLabels$.
The attributes of a tableau node $v$ are:
\begin{itemize}
	\item $\Type(v) \in \{\State, \NonState\}$, 
	
	\item $\SType(v) \in \{\Complex, \Simple\}$, called the {\em subtype} of $v$, 
	
	\item $\Label(v)$, which is a finite set of assertions or formulas, called the {\em label} of $v$,
	
	\item $\RFormulas(v)$, which is a finite set of so called {\em reduced assertions or formulas} of~$v$,
	
	\item $\Status(v) \in \{\Unexpanded$, $\Expanded$, $\Incomplete$, $\Blocked$, $\Unsat\} \cup \{\UnsatWrt(U)$ $\mid$ $U \subseteq V$ and,  for all $u \in U$, $\Type(u) = \State$ and $\SType(u) = \Complex\}$,
	
	\item $\AssSN(v)$, which is a finite set of so called {\em assertions suggested by nominals} for~$v$, available (i.e., $\neq \Null$) only when $\SType(v) = \Complex$ and $\Type(v) = \State$, and is non-empty only when $\Status(v) = \Incomplete$, 

	\item $\Repl(v) : \noms \to \noms$, which is a partial mapping specifying replacements of nominals for $v$, called in short the {\em nominal replacement} for $v$, and is available (i.e., $\neq \Null$) only when $\SType(v) = \Complex$.	
\myend
\end{itemize}
\end{Definition}
We define 
\begin{itemize}
	\item $\FullLabel(v) = \Label(v) \cup \RFormulas(v)$ if $\SType(v) = \Simple$, 
	\item $\FullLabel(v) = \Label(v) \cup \RFormulas(v) \cup \{a\!:\!b \mid \Repl(v)(b) = a\}$ otherwise. 
\end{itemize}

We call $v$ a {\em state} if $\Type(v) = \State$, and a {\em non-state} otherwise. A state is like an ``and''-node and a non-state is like an ``or''-node, when treating a tableau as an ``and-or'' graph. 
 
A node $v$ is called a {\em complex} node if $\SType(v) = \Complex$, and a {\em simple} node otherwise. 
The label of a complex node consists of assertions, while the label of a simple node consists of formulas. 
Using terminology of description logic, a complex node is like an ABox consisting of assertions about named individuals, while a simple node is like an unnamed individual and its label consists of properties of that individual. 
The root $\nu$ is a complex non-state with $\Label(\nu) = \Gamma$. 
The assertions/formulas in the label of a node $v$ are treated as requirements to be realized for~$v$. Realizing such requirements causes the graph to be expanded or modified. 

For the intuition behind $\RFormulas(v)$, consider an example situation when $\varphi \land \psi \in \Label(v)$. To realize the requirement $\varphi \land \psi$ for $v$, we can connect $v$ to a node $w$ that differs from $v$ in that $\Label(w)$ contains $\varphi$ and $\psi$ instead of $\varphi \land \psi$ and $\RFormulas(w) = \RFormulas(v) \cup \{\varphi \land \psi\}$. In general, $\RFormulas(v)$ contains assertions or formulas that have been reduced for~$v$. 

$\Status(v)$ is called the {\em status} of $v$. Possible statuses of nodes are: $\Unexpanded$, $\Expanded$, $\Incomplete$, $\Unsat$, $\Blocked$ and $\UnsatWrt(U)$, where $U$ is a set of complex states and $\UnsatWrt(U)$ is read as ``closed w.r.t.\ any node from $U$''. 
A node $v$ may have status $\Incomplete$ only when it is a complex state, and this status means that we would like to extend the label of $v$ with the assertions from $\AssSN(v)$ as one of the possibilities. 
Informally, $\Unsat$ means ``unsatisfiable'' and $\UnsatWrt(U)$ means ``unsatisfiable w.r.t.\ any node from $U$''. By $\UnsatWrt(\ldots)$ we denote $\UnsatWrt(U)$ for some $U$, and by $\UnsatWrtP{u}$ we denote $\UnsatWrt(U)$ for some $U$ containing~$u$. When negated, e.g., in the form $\neq \UnsatWrtP{u}$ or $\notin \{\UnsatWrtP{u}, \ldots\}$, we mean the considered status is different from $\UnsatWrt(U)$ for any $U$ that contains~$u$. 
A node may have status $\Blocked$ only when it is a simple node whose label contains some nominals. 
The status $\Blocked$ can be updated only to $\Unsat$ or $\UnsatWrt(\ldots)$.

A fact $\Repl(v)(b) = a$ means that the nominal $b$ has been replaced by $a$ for the complex node $v$. For example, this can be due to an assertion $a\!:\!b$. 

An edge departing from a node $v$ is labeled if and only if $v$ is a state. 
If $\tuple{v,w} \in E$, then we call $v$ a {\em predecessor} of $w$ and $w$ a {\em successor} of~$v$. 
Let the relation ``being an {\em ancestor}'' be the reflexive-transitive closure of the relation ``being a predecessor''. We say that $v$ is a {\em descendant} of $u$ if $u$ is an ancestor of~$v$. 

A tableau is constructed with {\em global caching} in the sense that, if $v_1$ and $v_2$ are different nodes, then $\Type(v_1) \neq \Type(v_2)$ or $\SType(v_1) \neq \SType(v_2)$ or $\Label(v_1) \neq \Label(v_2)$ or $\RFormulas(v_1) \neq \RFormulas(v_2)$ or $\Repl(v_1) \neq \Repl(v_2)$. 

Connecting a node $v$ to a successor, which is created if necessary, is done by Function~\ConToSucc($v$, $type$, $sType$, $label$, $reduced$, $nomRepl$, $eLabel$) on page~\pageref{proc: ConToSucc}, where the parameters $type$, $sType$, $label$, $reduced$, $nomRepl$ specify the attributes of the successor, and $eLabel$ stands for an edge label of the connection. By applying global caching, we first check whether an existing node can be used as such a successor of $v$. If not, a new node is created and used as a successor of $v$ by calling Function \NewSucc with the same parameters. 

\begin{figure*}[t!]
	\begin{function}[H]
		\caption{ConToSucc($v, type, sType, label, reduced, nomRepl, eLabel$)\label{proc: ConToSucc}}
		\GlobalData{a rooted graph $\tuple{V,E,\nu}$.}
		\Purpose{connect a node $v$ to a successor, which is created if necessary.}
		
		\uIf{there exists a node $w \in V$ such that $\Type(w) = type$, $\SType(w) = sType$, $\Label(w) = label$, $\RFormulas(w) = reduced$ and $\Repl(w) = nomRepl$}{
			$E := E \cup \{\tuple{v,w}\}$\; 
			\lIf{$\Type(v) = \State$}{add $eLabel$ to $\ELabels(v,w)$\;}
		}
		\Else{$w := \NewSucc(v, type, sType, label, reduced, nomRepl, eLabel)$\;}
		\Return{$w$}\;
	\end{function}
	
	\medskip
	
	\begin{function}[H]
		\caption{NewSucc($v, type, sType, label, reduced, nomRepl, eLabel$)\label{proc: NewSucc}}
		\GlobalData{a rooted graph $\tuple{V,E,\nu}$.}
		\Purpose{create a new successor for $v$.}
		\LinesNumberedHidden
		add a new node $w$ to $V$\;
		$\Type(w) := type$,\ 
		$\SType(w) := sType$,\ 
		$\Status(w) := \Unexpanded$\;
		$\Label(w) := label$,\ 
		$\RFormulas(w) := reduced$,\ 
		$\Repl(w) := nomRepl$\;
		
		\lIf{$sType = \Complex$ and $type = \State$}{$\AssSN(w) := \emptyset$\;}
		
		\uIf{$v = \Null$}{
			\lForEach{nominal $a$ occurring in $\Label(w)$}{$\Repl(w)(a) := a$\;}
		}
		\Else {
			$E := E \cup \{\tuple{v,w}\}$\;
			\lIf{$\Type(v) = \State$}{$\ELabels(v,w) := \{eLabel\}$\;}
		}
		\Return{$w$}\;
	\end{function}
\end{figure*}


\subsection{Tableau Rules}

Our tableau calculus \CHPDL consists of the following tableau rules:
\begin{itemize}
	\item the static rules for expanding a non-state, 
	\item the rule $\rReplNom$ for replacing nominals in a complex non-state,
	\item the rule $\rNom$ for dealing with nominals in a simple non-state,
	\item the rule $\rReexpand$ for re-expanding a complex non-state,
	\item the rule $\rFormingState$ for forming a state,
	\item the transitional rule $\rTrans$ for expanding a state, 
	\item the rule $\rUnsat$ for updating the status of a node to $\Unsat$ or $\UnsatWrt(\ldots)$.
\end{itemize}

The applicability of a rule to a tableau is explicitly specified for the static rules. For any of the other rules, we say that it is {\em applicable} to a tableau if its execution can make changes to the tableau. 

\subsubsection{The static rules for expanding a non-state}
The static rules are written downwards, with a set of assertions/formulas above the line as the {\em premise}, which represents the label of the node to which the rule is applied, and a number of sets of assertions/formulas below the line as the {\em (possible) conclusions}, which represent the labels of the successor nodes resulted from the application of the rule.
Possible conclusions of a static rule are separated by $\mid$. If a rule is unary (i.e., with only one possible conclusion), then its only conclusion is ``firm'' and we ignore the word ``possible''.
The meaning of a static rule is that, if the premise is satisfiable, then some of the possible conclusions are also satisfiable.

\begin{table}[t!]
	\[
	\begin{array}{|c|}
	\hline
	\ \\[-1.2ex]
	\rAnd\; \fracc{X, o\!:\!(\varphi \land \psi)}{X, o\!:\!\varphi, o\!:\!\psi} \hspace{8.3em}
	\rOr\; \fracc{X, o\!:\!(\varphi \lor \psi)}{X, o\!:\!\varphi \mid X, o\!:\!\psi} \\
	\ \\[-1.2ex]
	\hline
	\ \\[-1.2ex]
	\textrm{if } \alpha \notin \mindices,\; \alpha \textrm{ is not a~test, and } I_{\Aut_\alpha} = \{q_1,\ldots,q_k\}:\\
	\ \\[-1.2ex]
	\rAutB\; \fracc{X, o\!:\![\alpha]\varphi}{X, o\!:\![\Aut_\alpha,q_1]\varphi, \ldots, o\!:\![\Aut_\alpha,q_k]\varphi}\\
	\ \\[-1.2ex]
	\;\rAutD\; \fracc{X, o\!:\!\lDmd{\alpha}\varphi}{X, o\!:\!\lDmd{\Aut_\alpha,q_1}\varphi \mid \ldots \mid X, o\!:\!\lDmd{\Aut_\alpha,q_k}\varphi}\\
	\ \\[-1.2ex]
	\hline
	\ \\[-1.2ex]
	\textrm{if } \delta_A(q) = \{(\omega_1,q_1),\ldots,(\omega_k,q_k)\} \textrm{ and } q \notin F_A : \\
	\ \\[-1.2ex]
	\rBox\; \fracc{X, o\!:\![A,q]\varphi}{X, o\!:\![\omega_1][A,q_1]\varphi, \ldots, o\!:\![\omega_k][A,q_k]\varphi}\\
	\ \\[-1.2ex]
	\rDmd\; \fracc{X, o\!:\!\lDmd{A,q}\varphi}{X, o\!:\!\lDmd{\omega_1}\lDmd{A,q_1}\varphi \mid \ldots \mid X, o\!:\!\lDmd{\omega_k}\lDmd{A,q_k}\varphi}\\
	\ \\[-1.2ex]
	\hline
	\ \\[-1.2ex]
	\textrm{if } \delta_A(q) = \{(\omega_1,q_1),\ldots,(\omega_k,q_k)\} \textrm{ and } q \in F_A :\\
	\ \\[-1.2ex]
	\rBoxD\; \fracc{X, o\!:\![A,q]\varphi}{X, o\!:\![\omega_1][A,q_1]\varphi, \ldots, o\!:\![\omega_k][A,q_k]\varphi, o\!:\!\varphi}\\
	\ \\[-1.2ex]
	\rDmdD\; \fracc{X, o\!:\!\lDmd{A,q}\varphi}{X, o\!:\!\lDmd{\omega_1}\lDmd{A,q_1}\varphi \mid \ldots \mid X, o\!:\!\lDmd{\omega_k}\lDmd{A,q_k}\varphi \mid X, o\!:\!\varphi}\\
	\ \\[-1.2ex]
	\hline
	\ \\[-1.2ex]
	\rBoxQm\; \fracc{X, o\!:\![\psi?]\varphi}{X, o\!:\!\ovl{\psi} \mid X, o\!:\!\varphi}
	\hspace{8.3em}
	\rDmdQm\; \fracc{X, o\!:\!\lDmd{\psi?}\varphi}{X, o\!:\!\psi, o\!:\!\varphi}\\
	\ \\[-1.2ex]
	\hline
	\ \\[-1.2ex]
	\rBoxTrans\; \fracc{X, a\!:\![\sigma]\varphi, \sigma(a,b)}{X, a\!:\![\sigma]\varphi, \sigma(a,b), b\!:\!\varphi}
	\\[3ex]
	\hline
	\end{array}
	\]
	\caption{The static rules of \CHPDL.\label{table: cal}}
\end{table}

We use $X$ and $Y$ to denote sets of assertions/formulas and write $X, o\!:\!\varphi$ to denote $X \cup \{o\!:\!\varphi\}$ with the assumption that $o\!:\!\varphi \notin X$. The static rules of \CHPDL are specified in Table~\ref{table: cal} as schemas. For each of them, the distinguished assertions/formulas of the premise are called the {\em principal assertions/formulas} of the rule. 
A static rule $(\rho)$ as an instance of a schema given in Table~\ref{table: cal} is {\em applicable} to a node $v$ if the following conditions hold:
\begin{itemize}
	\item $\Status(v) = \Unexpanded$ and $\Type(v) = \NonState$, 
	\item the rules $\rReplNom$ and $\rNom$ are not applicable to~$v$, 
	\item the premise of the rule is equal to $\Label(v)$,
	\item the conditions accompanied with $(\rho)$ are satisfied,
	\item if $(\rho) \neq \rBoxTrans$, then the principal assertion/formula of $(\rho)$ does not belong to $\RFormulas(v)$, else the assertion $b\!:\!\varphi$ in the conclusion does not belong to $\FullLabel(v)$.
\end{itemize}

The last condition prevents applying the rule unnecessarily, because it has been applied to an ancestor node of $v$ that corresponds to the same state in the intended Kripke model.  

If $(\rho) \neq \rBoxTrans$ is a static rule applicable to $v$, then the application is as follows:
\begin{itemize}
	\item Let $\xi$ be the principal assertion/formula of $(\rho)$.
	\item Let $X_1, \ldots, X_k$ be the possible conclusions of $(\rho)$.
	\item For each $1 \leq i \leq k$, do $\ConToSucc(v$, $\NonState$, $\SType(v)$, $X_i$, $\RFormulas(v) \cup \{\xi\}$, $\Repl(v)$, $\Null)$, which is specified on page~\pageref{proc: ConToSucc}.\footnote{Here, $\Null$ is a constant standing for ``nothing'' as in C~programming, which means the parameter is not important for this case.}
	\item $\Status(v) := \Expanded$.
\end{itemize}

\noindent
If $\rBoxTrans$ is applicable to $v$, then the application is as follows:
\begin{itemize}
	\item Let $Y$ be the conclusion of $\rBoxTrans$.
	\item $\ConToSucc(v, \NonState, \SType(v), Y, \RFormulas(v), \Repl(v), \Null)$.
	\item $\Status(v) := \Expanded$.
\end{itemize}

Applying a static rule understood as a schema to a node $v$ means applying an instance of the schema to $v$. Such an instance is chosen as follows: choose assertions/formulas from $\Label(v)$ such that they can be ``unified'' with the principal assertions/formulas in the schema, then instantiate the schema by using the substitution resulted from that unification. 

\subsubsection{The rule $\rReplNom$ for replacing nominals}
%
%
%
If $\Status(v) = \Unexpanded$ and $\Label(v)$ contains $a\!:\!b$ with $a \neq b$ then:
    \begin{enumerate}
	\item let $X$ and $Y$ be the sets obtained from $\Label(v) - \{a\!:\!b\}$ and $\RFormulas(v)$, respectively, by replacing every occurrence of $b$ with~$a$, including the ones in automata of modal operators; 
	\item $w := \ConToSucc(v,\NonState,\Complex,X,Y,\Repl(v),\Null)$;
	\item $\Repl(w)(b) := a$;
	\item for each nominal $c$ such that $\Repl(v)(c) = b$, do $\Repl(w)(c) := a$;
	\item $\Status(v) := \Expanded$;
    \end{enumerate}

\subsubsection{The rule $\rNom$ for dealing with nominals}
If $\Status(v) \neq \Unsat$, $\Type(v) = \Simple$ and there exists $a \in \Label(v)$, then:
\begin{enumerate}
	\item for each complex state $u$ such that $\Status(v) \neq \UnsatWrtP{u}$ and $\Status(u) \neq \Incomplete$ and $v$ may affect the status of the root $\nu$ via a path through $u$,\footnote{That is, there exists a path from $\nu$ to $v$ via $u$ that does not contain any node with status $\Unsat$ or $\UnsatWrtP{u}$.} do:
	\begin{enumerate}
		\item $X := \{a\!:\!\varphi \mid \varphi \in \Label(v)$ and $\varphi \neq a\}$;
		\item\label{step: UDJWS} if there exists $\xi \in X$ such that $\ovl{\xi} \in \FullLabel(u)$ then 
			\begin{enumerate}
				\item if $\Status(v)$ is of the form $\UnsatWrt(U)$ then\\ $\Status(v) := \UnsatWrt(U \cup \{u\})$, 
				\item else $\Status(v) := \UnsatWrt(\{u\})$; 
			\end{enumerate}
		\item else if $X \nsubseteq \FullLabel(u)$ then 
			\begin{enumerate}
				\item $\Status(u) := \Incomplete$; 
				\item $\AssSN(u) := X - \FullLabel(u)$;
			\end{enumerate}
	\end{enumerate}
	\item if $\Status(v) = \Unexpanded$ then $\Status(v) := \Blocked$.
\end{enumerate}

\subsubsection{The rule $\rReexpand$ for re-expanding a complex non-state}
If $\tuple{v,w} \in E$ and $\Status(w) = \Incomplete$ then:\\
(we must have that $\SType(v) = \Complex$ and $\Type(v) = \NonState$)
\begin{enumerate}
\item\label{step: JPEJS} delete the edge $\tuple{v,w}$ from $E$;
\item $X := \Label(v) \cup \AssSN(w)$;
\item $\ConToSucc(v,\NonState,\SType(v),X,\RFormulas(v),\Repl(v),\Null)$;
\item for each $\xi \in \AssSN(w)$, do 
	$\ConToSucc(v$, $\NonState$, $\SType(v)$, $\Label(v) \cup \{\ovl{\xi}\}$, $\RFormulas(v)$, $\Repl(v)$, $\Null)$.
\end{enumerate}

\subsubsection{The rule $\rFormingState$ for forming a state}
If $\Status(v) = \Unexpanded$, $\Type(v) = \NonState$ and no rule among the static rules, $\rReplNom$ and $\rNom$ is applicable to $v$, then:
\begin{enumerate}
	\item if $\SType(v) = \Complex$ then\\
		$\ConToSucc(v,\State,\Complex,\Label(v),\RFormulas(v),\Repl(v),\Null)$,
	\item else $\ConToSucc(v,\State,\Simple,\Label(v),\emptyset,\Null,\Null)$;
	\item $\Status(v) := \Expanded$.
\end{enumerate}

We need the set $\RFormulas(v)$ for the successor of $v$ in the case $\SType(v) = \Complex$ due to the situation related with the rule $\rNom$. 

\subsubsection{The transitional rule $\rTrans$ for expanding a state}
If $\Type(v) = \State$ and $\Status(v) = \Unexpanded$ then:
\begin{enumerate}
\item for each $o\!:\!\lDmd{\sigma}\varphi \in \Label(v)$ do
	\begin{enumerate}
		\item $X := \{\varphi\} \cup \{\psi \mid o\!:\![\sigma]\psi \in \Label(v)\}$;
		\item $\ConToSucc(v,\NonState,\Simple,X,\emptyset,\Null,o\!:\!\lDmd{\sigma}\varphi)$;
	\end{enumerate}	
\item $\Status(v) := \Expanded$.
\end{enumerate}


\subsubsection{The rule $\rUnsat$ for updating the status of a node}
We need the following definition before specifying the rule $\rUnsat$.

\begin{Definition}\label{def: YUDSS}
	Let $\xi \in \FullLabel(v)$ be of the form $o\!:\!\lDmd{A,q}\varphi$ or $o\!:\!\lDmd{\omega}\lDmd{A,q}\varphi$. We say that $\xi$ is {\em $\Dmd$-realizable at $v$ w.r.t.~$u$}, where $u$ is a complex state and an ancestor of $v$ such that $\Status(u) \notin \{\Unsat,\Incomplete\}$, iff the following conditions hold:
	\begin{enumerate}
		\item $\Status(v) \notin \{\Unsat$, $\UnsatWrtP{u}\}$,
		\item either $\Status(v) \in \{\Unexpanded,\Incomplete\}$ or some of the following conditions hold:
		\begin{enumerate}
			\item\label{item: YUDSS a} $\xi = o\!:\!\lDmd{A,q}\varphi$, $q \in F_A$ and $o\!:\!\varphi \in \FullLabel(v)$;
			\item\label{item: YUDSS b} $\xi = o\!:\!\lDmd{A,q}\varphi$, $(q,\omega,q') \in \delta_A$, $o\!:\!\lDmd{\omega}\lDmd{A,q'}\varphi$ belongs to $\FullLabel(v)$ and is $\Dmd$-realizable at $v$ w.r.t.~$u$; 
			\item\label{item: YUDSS c} $\xi = o\!:\!\lDmd{\psi?}\lDmd{A,q}\varphi$, $\{o\!:\!\psi,o\!:\!\lDmd{A,q}\varphi\} \subseteq \FullLabel(v)$ and $o\!:\!\lDmd{A,q}\varphi$ is $\Dmd$-realizable at $v$ w.r.t.~$u$; 
			
			\item\label{item: YUDSS e} $v$ is expanded by the tableau rule $\rDmdD$, $\xi = o\!:\!\lDmd{A,q}\varphi$ is the principal assertion/formula, and the successor $w$ of $v$ whose label is obtained from $\Label(v)$ by replacing $\xi$ with $o\!:\!\varphi$ has $\Status(w) \notin \{\Unsat,\UnsatWrtP{u}\}$;
			
			\item\label{item: YUDSS f} $v$ is expanded by the tableau rule $\rDmd$ or $\rDmdD$, $\xi = o\!:\!\lDmd{A,q}\varphi$ is the principal assertion/formula, $w$ is a successor of~$v$ and the assertion/formula $\xi' \in \Label(w)$ obtained from $\xi$ is $\Dmd$-realizable at $w$ w.r.t.~$u$;
			
			\item\label{item: YUDSS g} $v$ is expanded by the tableau rule $\rDmdQm$, $\xi = o\!:\!\lDmd{\psi?}\lDmd{A,q}\varphi$ is the principal assertion/formula, $w$ is the unique successor of $v$, and $o\!:\!\lDmd{A,q}\varphi$ is $\Dmd$-realizable at $w$ w.r.t.~$u$; 
			
			\item\label{item: YUDSS h} $v$ is expanded by the rule $\rReexpand$, $\rFormingState$ or a static tableau rule and $\xi$ is not a principal assertion/formula, and $\xi$ is $\Dmd$-realizable at a successor of $v$ w.r.t.~$u$;
			
			\item\label{item: YUDSS j} $v$ is expanded by the rule $\rTrans$, $\xi = o\!:\!\lDmd{\sigma}\lDmd{A,q}\varphi$, $w$ is a successor of $v$, $\xi \in \ELabels(v,w)$, $\lDmd{A,q}\varphi \in \Label(w)$, and $\lDmd{A,q}\varphi$ is $\Dmd$-realizable at~$w$ w.r.t.~$u$;
			
			\item\label{item: YUDSS k} there exists $a \in \Label(v)$ such that $a\!:\!\xi$ is $\Dmd$-realizable at~$u$ w.r.t.~$u$.
			\ \ \ \myEnd
		\end{enumerate}
	\end{enumerate} 
\end{Definition}

Observe that the notion of $\Dmd$-realizability is defined inductively and the conditions~\ref{item: YUDSS a} and~\ref{item: YUDSS e} correspond to the base cases.

The rule $\rUnsat$ is specified as follows: 
\begin{enumerate}
\item If $\Status(v) \neq \Unsat$ then:
    \begin{enumerate}
	\item if (there exists $o\!:\!\bot \in \Label(v)$ or $a\!:\!\lnot a \in \Label(v)$ or $\{\xi,\ovl{\xi}\} \subseteq \FullLabel(v)$)\\
	or $\Status(v) = \UnsatWrtP{v}$, 
	then $\Status(v) := \Unsat$,\footnote{As an optimization, if there exists $\xi \in \FullLabel(v)$ of the form $o\!:\!\lDmd{A,q}\varphi$ or $o\!:\!\lDmd{\omega}\lDmd{A,q}\varphi$ such that it is not $\Dmd$-realizable at $v$ w.r.t.\ a complex state $u$ that is an ancestor of $v$ and checking $\Dmd$-realizability of $\xi$ at $v$ w.r.t.~$u$ does not go through any simple non-state whose label contains a nominal, then $\Dmd$-nonrealizability of $\xi$ at $v$ does not really depend on $u$, and $\Status(v)$ can be changed to $\Unsat$.}
	\item else if $u$ is a complex state and an ancestor of $v$ such that $\Status(u) \notin \{\Unsat$, $\Incomplete\}$ and there exists $\xi \in \FullLabel(v)$ of the form $o\!:\!\lDmd{A,q}\varphi$ or $o\!:\!\lDmd{\omega}\lDmd{A,q}\varphi$ that is not $\Dmd$-realizable at $v$ w.r.t.~$u$ and $\Status(v) \neq \UnsatWrtP{u}$, then:
	    \begin{enumerate}
	   	\item if $\Status(v)$ is of the form $\UnsatWrt(U)$,\\ then $\Status(v) := \UnsatWrt(U \cup \{u\})$,
	   	\item else $\Status(v) := \UnsatWrt(\{u\})$.
	    \end{enumerate}
	\end{enumerate}
	
\item If $\Status(v) \notin \{\Unexpanded,\Unsat,\Blocked,\UnsatWrt(\ldots)\}$ and $\Type(v) = \NonState$, then:
  \begin{enumerate}
  	\item\label{item: JHCDS} if all successors of $v$ have status $\Unsat$ then $\Status(v) := \Unsat$,
  	\item\label{item: JHDWA} else if every successor of $v$ has status $\Unsat$ or $\UnsatWrt(\ldots)$ then:
  	\begin{enumerate}
  		\item let $w_1,\ldots,w_k$ be all the successors of $v$ such that, for $1 \leq i \leq k$, $\Status(w_i)$ is of the form $\UnsatWrt(U_i)$, and let $U = \bigcap_{1 \leq i \leq k} U_i$;
  		\item if $U \neq \emptyset$ then: if $\Status(v)$ is of the form $\UnsatWrt(U')$, then $\Status(v) := \UnsatWrt(U' \cup U)$, else $\Status(v) := \UnsatWrt(U)$.
  	\end{enumerate}
  \end{enumerate}
  \item If $\Status(v) \notin \{\Unexpanded,\Unsat,\Incomplete\}$ and $\Type(v) = \State$, then:
  \begin{enumerate}
  	\item\label{item: JIRSA} if $v$ has a successor $w$ with $\Status(w) = \Unsat$, then $\Status(v) := \Unsat$,
  	\item\label{item: JHREA} else if $v$ has a successor $w$ with $\Status(w) = \UnsatWrt(U)$ and $\Status(v)$ is not of the form $\UnsatWrt(U')$ with $U' \supseteq U$, then:
  	\begin{enumerate}
  		\item if $\Status(v)$ is of the form $\UnsatWrt(U')$,\\
  		then $\Status(v) := \UnsatWrt(U' \cup U)$, 
  		\item else $\Status(v) := \UnsatWrt(U)$.
  	\end{enumerate}
  \end{enumerate}
\end{enumerate}


\subsection{Checking Satisfiability}

Let $\Gamma$ be an ABox in NNF. A {\em \CHPDL-tableau} for $\Gamma$ is a tableau $G = (V,E,\nu)$ constructed as follows. At the beginning, $V := \emptyset$, $E := \emptyset$ and 
$\nu := \NewSucc(\Null$, $\NonState$, $\Complex$, $\Gamma$, $\emptyset$, $\emptyset$, $\Null)$.
Then, while $\Status(\nu) \neq \Unsat$ and there is some tableau rule $(\rho)$ applicable to some node $v$, choose such a pair $((\rho),v)$ and apply $(\rho)$ to $v$.\footnote{As an optimization, it makes sense to expand $v$ only when there may exist a path from the root to $v$ that does not contain any node with the status $\Unsat$.} Observe that the set of all assertions and formulas that may appear in the contents of the nodes of $G$ is finite. Due to global caching, $G$ is finite and can be effectively constructed. 
\LongVersion{The following theorem immediately follows from Corollaries~\ref{cor: Soundness} and~\ref{cor: Completeness}, which are given and proved in Section~\ref{section: Proofs}.}

\begin{theorem}[Soundness and Completeness]\label{theorem: HJDFS}
	Let $\Gamma$ be an ABox in NNF and $G = (V,E,\nu)$ an arbitrary \CHPDL-tableau for~$\Gamma$. Then, $\Gamma$ is satisfiable if and only if $\Status(\nu) \neq \Unsat$.
	\myEnd
\end{theorem}

To check satisfiability of an ABox $\Gamma$ in NNF, one can construct a \CHPDL-tableau $G = (V,E,\nu)$ for~$\Gamma$ and return ``no'' when $\Status(\nu) = \Unsat$, or ``yes'' otherwise. We call this the {\em \CHPDL-tableau decision procedure}.
%
%
The following corollary immediately follows from Corollary~\ref{cor: JDKSA}, which is given and proved in Section~\ref{section: Proofs}.

\begin{corollary}\label{cor: UDEKS}
	The \CHPDL-tableau decision procedure has the \EXPTIME complexity.
	\myEnd
\end{corollary}


\section{An Illustrative Example}
\label{sec: example}

Consider the following ABox in NNF:
\[ \Gamma = \{a\!:\![\sigma^*]p,\; \sigma(a,b),\; b\!:\!\lDmd{(a? \cup \sigma)^*}\lnot p\}. \]
We show that $\Gamma$ is unsatisfiable by constructing a \CHPDL-tableau $G = (V,E,\nu)$ for $\Gamma$ with $\Status(\nu) = \Unsat$. In the construction, we use the following finite automata:
\[
\begin{array}{rcl}
A_1 = & \Aut_{\sigma^*} & = (\{\sigma\},\{0\},\{0\},\{(0,\sigma,0)\},\{0\}),\\
A_2 = & \Aut_{(a?\, \cup \sigma)^*} & = (\{\sigma,a?\},\{0\},\{0\},\{(0,\sigma,0),(0,a?,0)\},\{0\}),\\
A_3 = & \Aut_{(b?\, \cup \sigma)^*} & = (\{\sigma,b?\},\{0\},\{0\},\{(0,\sigma,0),(0,b?,0)\},\{0\}).
\end{array}
\]
As these automata have only one state (named 0), we will write $A_i$ instead of $(A_i,0)$, for $1 \leq i \leq 3$. 
For example, $[A_1]$ and $\lDmd{A_1}$ stand for $[A_1,0]$ and $\lDmd{A_1,0}$, respectively. 
The constructed \CHPDL-tableau $G$ is illustrated in Figure~\ref{fig: example}. 

\begin{figure}[h!]
\ramka{
\begin{center}
\begin{tabular}{c}
\xymatrix{
\nu
\ar@{-->}[d]
& v_8
\ar@{->}[r]
& v_{27}
\ar@{->}[r]
& v_{28}
\ar@{->}[r]
\ar@{->}[rd]
\ar@{->}[d]
& v_{31}
\ar@{->}[r]
& v_{39}
\ar@{->}[d]
\\
v_5
\ar@{->}[ru]
\ar@{->}[r]
\ar@{->}[d]
& v_7
\ar@{->}[r]
\ar@{->}[d]
\ar@{.>}[ldd]
& v_{18}
\ar@{->}[rd]
& v_{29}
& v_{30}
\ar@{->}[d]
& *+[F]{v_{40}}
\\
v_6
& v_{17}
\ar@{->}[r]
\ar@{->}[rd]
\ar@{->}[d]
& v_{19}
& v_{25}
\ar@{->}[d]
& *+[F]{v_{32}}
\ar@{->}[d]
\\
*+[F]{v_9}
\ar@{->}[d]
& v_{20}
\ar@{->}[d]
& v_{21}
\ar@{->}[dr]
& v_{26}
& v_{33}
\ar@{->}[d]
& *+[F]{v_{37}}
\ar@{->}[l]
\\
v_{10}
\ar@{->}[d]
& *+[F]{v_{22}}
\ar@{->}[l]
& *+[F]{v_{24}}
\ar@/^{0.8pc}/@{->}[ll]
& v_{23}
\ar@{->}[l]
& v_{34}
\ar@{->}[r]
\ar@{->}[d]
\ar@/^{3.8pc}/@{->}[lllldd]
& v_{35}
\ar@{->}[u]
\\
v_{11}
\ar@{->}[r]
\ar@{->}[rd]
\ar@{->}[d]
& v_{13}
\ar@{->}[r]
& *+[F]{v_{15}}
\ar@{->}[ull]
& {}
& v_{36}
\ar@{->}[d]
\\
v_{12}
& v_{14}
\ar@{->}[r]
& v_{16}
& {}
& v_{38}
} 
\end{tabular}
\end{center}
} 
\caption{An illustration for Section~\ref{sec: example}. The dashed edge from $\nu$ to $v_5$ represents a path \mbox{$\nu \to v_1 \to \ldots \to v_5$}. The dotted edge from $v_7$ to $v_9$ stands for a normal edge before re-expanding $v_7$. By the re-expansion, that edge is deleted and $v_7$ is connected to the newly created nodes $v_{17}$ and $v_{18}$. The nodes in rectangular frames are states, the others are non-states. The nodes $v_{10}$ -- $v_{16}$ and $v_{33}$ -- $v_{38}$ are simple nodes, the others are complex nodes. The nodes $v_{16}$ and $v_{38}$ have the status $\Blocked$.\label{fig: example}}
\end{figure}
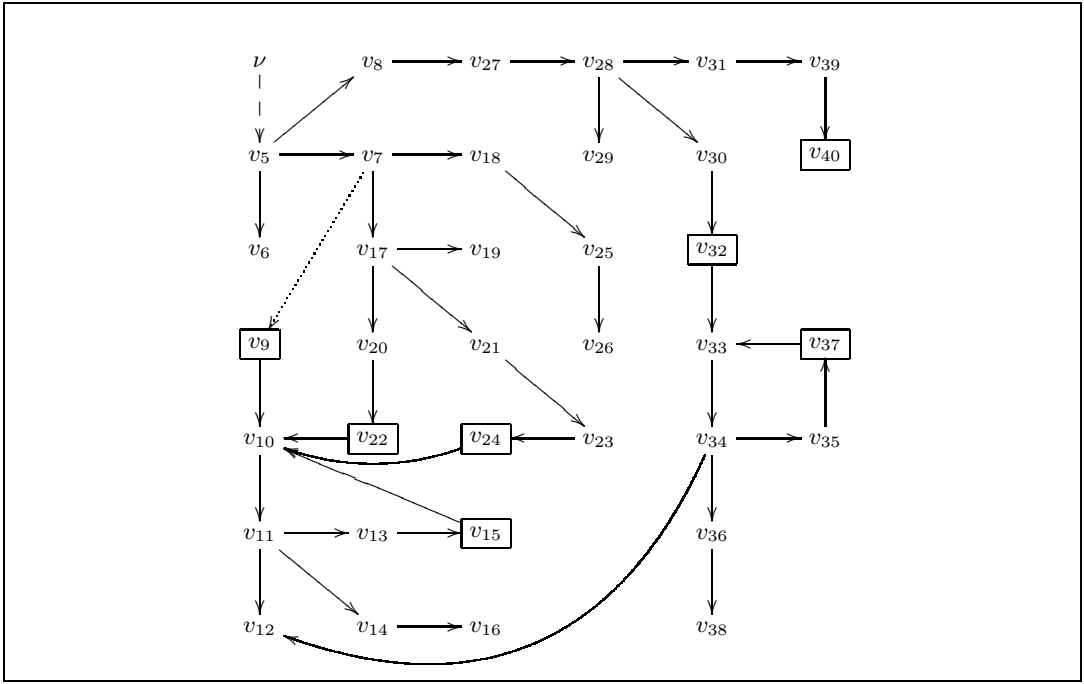

At the beginning, $G$ contains the root $\nu$ with $\Label(\nu) = \Gamma$.

Applying $\rAutB$, $\nu$ is connected to a new complex non-state $v_1$ with
\[ \Label(v_1) = \{a\!:\![A_1]p,\; \sigma(a,b),\; b\!:\!\lDmd{(a? \cup \sigma)^*}\lnot p\}. \]

Applying $\rBoxD$, $v_1$ is connected to a new complex non-state $v_2$ with
\[ \Label(v_2) = \{a\!:\![\sigma][A_1]p,\; a\!:\!p,\; \sigma(a,b),\; b\!:\!\lDmd{(a? \cup \sigma)^*}\lnot p\}. \]

Applying $\rBoxTrans$, $v_2$ is connected to a new complex non-state $v_3$ with
\[ \Label(v_3) = \Label(v_2) \cup \{b\!:\![A_1]p\}. \]

Applying $\rBoxD$, $v_3$ is connected to a new complex non-state $v_4$ with
\[ \Label(v_4) = \Label(v_2) \cup \{b\!:\![\sigma][A_1]p,\; b\!:\!p\}. \]

Applying $\rAutD$, $v_4$ is connected to a new complex non-state $v_5$ with
\[ \Label(v_5) = \{a\!:\![\sigma][A_1]p,\; a\!:\!p,\; \sigma(a,b),\; b\!:\![\sigma][A_1]p,\; b\!:\!p,\; b\!:\!\lDmd{A_2}\lnot p\}. \]

Applying $\rDmdD$, $v_5$ is connected to new complex non-states $v_6$, $v_7$, and $v_8$ with
\begin{eqnarray*}
	\Label(v_6) & = & \{a\!:\![\sigma][A_1]p,\; a\!:\!p,\; \sigma(a,b),\; b\!:\![\sigma][A_1]p,\; b\!:\!p,\; b\!:\!\lnot p\},\\
	\Label(v_7) & = & \{a\!:\![\sigma][A_1]p,\; a\!:\!p,\; \sigma(a,b),\; b\!:\![\sigma][A_1]p,\; b\!:\!p,\; b\!:\!\lDmd{\sigma}\lDmd{A_2}\lnot p\},\\
	\Label(v_8) & = & \{a\!:\![\sigma][A_1]p,\; a\!:\!p,\; \sigma(a,b),\; b\!:\![\sigma][A_1]p,\; b\!:\!p,\; b\!:\!\lDmd{a?}\lDmd{A_2}\lnot p\}.
\end{eqnarray*}

Applying $\rUnsat$, $\Status(v_6)$ is changed to $\Unsat$.

Applying $\rFormingState$, $v_7$ is connected to a new complex state $v_9$ with $\Label(v_9) = \Label(v_7)$. 
Applying $\rTrans$, $v_9$ is connected to a new simple non-state $v_{10}$ with 
\begin{eqnarray*}
	\Label(v_{10}) & = & \{\lDmd{A_2}\lnot p,\; [A_1]p\},\\
	\ELabels(v_9, v_{10}) & = & \{b\!:\!\lDmd{\sigma}\lDmd{A_2}\lnot p\}.
\end{eqnarray*}

Applying $\rBoxD$, $v_{10}$ is connected to a new simple non-state $v_{11}$ with  
\[ \Label(v_{11}) = \{\lDmd{A_2}\lnot p,\; [\sigma][A_1]p,\; p\}. \]

Applying $\rDmdD$, $v_{11}$ is connected to new simple non-states $v_{12}$, $v_{13}$, and $v_{14}$ with  
\begin{eqnarray*}
	\Label(v_{12}) & = & \{\lnot p,\; [\sigma][A_1]p,\; p\},\\
	\Label(v_{13}) & = & \{\lDmd{\sigma}\lDmd{A_2}\lnot p,\; [\sigma][A_1]p,\; p\},\\
	\Label(v_{14}) & = & \{\lDmd{a?}\lDmd{A_2}\lnot p,\; [\sigma][A_1]p,\; p\}.
\end{eqnarray*}

Applying $\rUnsat$, $\Status(v_{12})$ is changed to $\Unsat$.

Applying $\rFormingState$, $v_{13}$ is connected to a new simple state $v_{15}$ with $\Label(v_{15}) = \Label(v_{13})$. Applying $\rTrans$, $v_{15}$ is connected to the existing node $v_{10}$ with $\ELabels(v_{15}, v_{10}) = \{\lDmd{\sigma}\lDmd{A_2}\lnot p\}$.

Applying $\rDmdQm$, $v_{14}$ is connected to a new simple non-state $v_{16}$ with 
\[ \Label(v_{16}) = \{a,\; \lDmd{A_2}\lnot p,\; [\sigma][A_1]p,\; p\}. \]

Applying $\rNom$ to $v_{16}$ and the complex state $v_9$, $\Status(v_9)$ is changed to $\Incomplete$, $\AssSN(v_9)$ is set to $\{a\!:\!\lDmd{A_2}\lnot p\}$, and $\Status(v_{16})$ is changed to $\Blocked$. 

Applying $\rReexpand$ to $v_7$ and the incomplete complex state $v_9$, the edge $(v_7,v_9)$ is deleted and $v_7$ is connected to new complex non-states $v_{17}$ and $v_{18}$ with 
\begin{eqnarray*}
	\Label(v_{17}) & = & \Label(v_7) \cup \{a\!:\!\lDmd{A_2}\lnot p\},\\
	\Label(v_{18}) & = & \Label(v_7) \cup \{a\!:\![A_2]p\}.
\end{eqnarray*}

Applying $\rDmdD$, $v_{17}$ is connected to new complex non-states $v_{19}$, $v_{20}$, $v_{21}$ with
\begin{eqnarray*}
	\Label(v_{19}) & = & \Label(v_7) \cup \{a\!:\!\lnot p\},\\
	\Label(v_{20}) & = & \Label(v_7) \cup \{a\!:\!\lDmd{\sigma}\lDmd{A_2}\lnot p\},\\
	\Label(v_{21}) & = & \Label(v_7) \cup \{a\!:\!\lDmd{a?}\lDmd{A_2}\lnot p\}.
\end{eqnarray*}

Applying $\rUnsat$, $\Status(v_{19})$ is changed to $\Unsat$ (due to $a\!:\!p$ and $a\!:\!\lnot p$).

Applying $\rFormingState$, $v_{20}$ is connected to a new complex state $v_{22}$ with $\Label(v_{22}) = \Label(v_{20})$. Applying $\rTrans$, $v_{22}$ is connected to the existing node $v_{10}$ with $\ELabels(v_{22}, v_{10}) = \{a\!:\!\lDmd{\sigma}\lDmd{A_2}\lnot p,\; b\!:\!\lDmd{\sigma}\lDmd{A_2}\lnot p\}$.

Observe that the assertion $b\!:\!\lDmd{\sigma}\lDmd{A_2}\lnot p$ is not $\Dmd$-realizable at $v_{22}$ w.r.t.\ $v_{22}$. Hence, applying $\rUnsat$, $\Status(v_{22})$ is first changed to $\UnsatWrt(\{v_{22}\})$ and then to $\Unsat$, after that $\Status(v_{20})$ is also changed to $\Unsat$. 

Applying $\rDmdQm$, $v_{21}$ is connected to a new complex non-state $v_{23}$ with 
\[ \Label(v_{23}) = \Label(v_7) \cup \{a\!:\!a,\; a\!:\!\lDmd{A_2}\lnot p\}. \]

Notice that $a\!:\!\lDmd{A_2}\lnot p \in \RFormulas(v_{23})$. 
Applying $\rFormingState$, $v_{23}$ is connected to a new complex state $v_{24}$ with $\Label(v_{24}) = \Label(v_{23})$. Applying $\rTrans$, $v_{24}$ is connected to the existing node $v_{10}$ with $\ELabels(v_{24}, v_{10}) = \{b\!:\!\lDmd{\sigma}\lDmd{A_2}\lnot p\}$.

Observe that the assertion $b\!:\!\lDmd{\sigma}\lDmd{A_2}\lnot p$ is not $\Dmd$-realizable at $v_{24}$ w.r.t.\ $v_{24}$. Hence, applying $\rUnsat$, $\Status(v_{24})$ is first changed to $\UnsatWrt(\{v_{24}\})$ and then to $\Unsat$, after that the statuses of $v_{23}$, $v_{21}$ and $v_{17}$ are changed to $\Unsat$ in subsequent steps.

Applying $\rBoxD$, $v_{18}$ is connected to a new complex non-state $v_{25}$ with
\[ \Label(v_{25}) = \Label(v_7) \cup \{a\!:\![\sigma][A_2]p,\; a\!:\![a?][A_2]p\}. \]

Applying $\rBoxTrans$, $v_{25}$ is connected to new complex non-states $v_{26}$ with
\[ \Label(v_{26}) = \Label(v_{25}) \cup \{b\!:\![A_2]p\}. \]

Applying $\rUnsat$, $\Status(v_{26})$ is changed to $\Unsat$ (due to $b\!:\!\lDmd{A_2}\lnot p$ and $b\!:\![A_2]p$), and then the statuses of $v_{25}$, $v_{18}$ and $v_7$ are changed to $\Unsat$ in subsequent steps. 

Applying $\rDmdQm$, $v_8$ is connected to a new complex non-state $v_{27}$ with 
\[ \Label(v_{27}) = \{a\!:\![\sigma][A_1]p,\; a\!:\!p,\; \sigma(a,b),\; b\!:\![\sigma][A_1]p,\; b\!:\!p,\; b\!:\!a,\; b\!:\!\lDmd{A_2}\lnot p\}. \]

Applying $\rReplNom$, $v_{27}$ is connected to a new complex non-state $v_{28}$ with 
\[ \Label(v_{28}) = \{b\!:\![\sigma][A_1]p,\; b\!:\!p,\; \sigma(b,b),\; b\!:\!\lDmd{A_3}\lnot p\}. \]

Applying $\rDmdD$, $v_{28}$ is connected to new complex non-states $v_{29}$, $v_{30}$, $v_{31}$ with
\begin{eqnarray*}
	\Label(v_{29}) & = & \{b\!:\![\sigma][A_1]p,\; b\!:\!p,\; \sigma(b,b),\; b\!:\!\lnot p\},\\
	\Label(v_{30}) & = & \{b\!:\![\sigma][A_1]p,\; b\!:\!p,\; \sigma(b,b),\; b\!:\!\lDmd{\sigma}\lDmd{A_3}\lnot p\},\\
	\Label(v_{31}) & = & \{b\!:\![\sigma][A_1]p,\; b\!:\!p,\; \sigma(b,b),\; b\!:\!\lDmd{b?}\lDmd{A_3}\lnot p\}.
\end{eqnarray*}

Applying $\rUnsat$, $\Status(v_{29})$ is changed to $\Unsat$.

Notice that $b\!:\![A_1]p \in \RFormulas(v_{30})$. Applying $\rFormingState$, $v_{30}$ is connected to a new complex state $v_{32}$ with $\Label(v_{32}) = \Label(v_{30})$. Applying $\rTrans$, $v_{32}$ is connected to a new simple non-state $v_{33}$ with 
\begin{eqnarray*}
	\Label(v_{33}) & = & \{\lDmd{A_3}\lnot p,\; [A_1]p\},\\
	\ELabels(v_{32},v_{33}) & = & \{b\!:\!\lDmd{\sigma}\lDmd{A_3}\lnot p\}.
\end{eqnarray*}

Applying $\rBoxD$, $v_{33}$ is connected to a new simple non-state $v_{34}$ with
\[ \Label(v_{34}) = \{\lDmd{A_3}\lnot p,\; [\sigma][A_1]p,\; p\}. \]

Applying $\rDmdD$, $v_{34}$ is connected to the existing node $v_{12}$ and new simple non-states $v_{35}$ and $v_{36}$ with 
\begin{eqnarray*}
	\Label(v_{35}) & = & \{\lDmd{\sigma}\lDmd{A_3}\lnot p,\; [\sigma][A_1]p,\; p\},\\
	\Label(v_{36}) & = & \{\lDmd{b?}\lDmd{A_3}\lnot p,\; [\sigma][A_1]p,\; p\}.
\end{eqnarray*}

Applying $\rFormingState$, $v_{35}$ is connected to a new simple state $v_{37}$ with $\Label(v_{37}) = \Label(v_{35})$. Applying $\rTrans$, $v_{37}$ is connected to the existing node $v_{33}$ with $\ELabels(v_{37},v_{33}) = \{\lDmd{\sigma}\lDmd{A_3}\lnot p\}$.

Applying $\rDmdQm$, $v_{36}$ is connected to a new simple non-state $v_{38}$ with
\[ \Label(v_{38}) = \{b,\; \lDmd{A_3}\lnot p,\; [\sigma][A_1]p,\; p\}. \]

Applying $\rNom$, $\Status(v_{38})$ is changed to $\Blocked$. Notice that $v_{38}$ is ``compatible'' with $v_{32}$, which is the only complex state that is an ancestor of $v_{38}$, and hence, $\Status(v_{32})$ is not changed (to $\Incomplete$). 

Observe that the assertion $b\!:\!\lDmd{\sigma}\lDmd{A_3}\lnot p$ is not $\Dmd$-realizable at $v_{32}$ w.r.t.\ $v_{32}$. Hence, applying $\rUnsat$, $\Status(v_{32})$ is first changed to $\UnsatWrt(\{v_{32}\})$ and then to $\Unsat$, after that $\Status(v_{30})$ is also changed to $\Unsat$. 

Applying $\rDmdQm$, $v_{31}$ is connected to a new simple non-state $v_{39}$ with 
\[ \Label(v_{39}) = \{b\!:\![\sigma][A_1]p,\; b\!:\!p,\; \sigma(b,b),\; b\!:\!b,\; b\!:\!\lDmd{A_3}\lnot p\}. \]

Notice that $b\!:\!\lDmd{A_3}\lnot p \in \RFormulas(v_{39})$. Applying $\rFormingState$, $v_{39}$ is connected to a new complex state $v_{40}$ with $\Label(v_{40}) = \Label(v_{39})$. Applying $\rTrans$, $\Status(v_{40})$ is changed to $\Expanded$ (without being connected to any nodes).

Observe that the assertion $b\!:\!\lDmd{A_3}\lnot p$ is not $\Dmd$-realizable at $v_{40}$ w.r.t.\ $v_{40}$. Hence, applying $\rUnsat$, $\Status(v_{40})$ is first changed to $\UnsatWrt(\{v_{40}\})$ and then to $\Unsat$, after that the statuses of $v_{39}$, $v_{31}$, $v_{28}$, $v_{27}$, $v_8$, $v_5$ -- $v_1$, $\nu$ are changed to $\Unsat$ in subsequent steps.
Since $\Status(\nu) = \Unsat$, by Theorem~\ref{theorem: HJDFS}, we conclude that the given ABox $\Gamma$ is unsatisfiable. 


\section{Proofs}
\label{section: Proofs}
	
In this section, let $\Gamma$ be an ABox in NNF and $G = (V,E,\nu)$ an arbitrary \CHPDL-tableau for~$\Gamma$. 
	
\subsection{Complexity Analysis}
\label{subsection: complexity}

We define the {\em length} of a formula (resp.\ program or assertion) to be the number of occurrences of symbols in that formula (resp.\ program or assertion). We define the {\em size} of a set of formulas and assertions to be the sum of the lengths of its formulas and assertions. 

\begin{Definition}\label{def: HJSAP}
	The set of {\em basic subformulas} of $\Gamma$, denoted by $\bsf(\Gamma)$, consists of all subformulas of $\Gamma$ and their negations in NNF. The set $\clsZ(\Gamma)$ is defined to be the smallest extension of $\Gamma \cup \bsf(\Gamma)$ such that:
	\begin{enumerate}
		\item if $[\alpha]\varphi \in \bsf(\Gamma)$, $q \in Q_{\Aut_\alpha}$, $\omega$ is of the form $\sigma$ or $\psi?$ and occurs in $\alpha$, then $[\Aut_\alpha,q]\varphi$ and $[\omega][\Aut_\alpha,q]\varphi$ belong to $\clsZ(\Gamma)$;
		\item if $\lDmd{\alpha}\varphi \in \bsf(\Gamma)$, $q \in Q_{\Aut_\alpha}$, $\omega$ is of the form $\sigma$ or $\psi?$ and occurs in $\alpha$, then $\lDmd{\Aut_\alpha,q}\varphi$ and $\lDmd{\omega}\lDmd{\Aut_\alpha,q}\varphi$ belong to $\clsZ(\Gamma)$;
		\item if $\varphi \in \clsZ(\Gamma)$ and $a$ is a nominal occurring in $\Gamma$, then $a\!:\!\varphi \in \clsZ(\Gamma)$.
	\end{enumerate}
	The set $\cls(\Gamma)$ is defined to be the smallest extension of $\clsZ(\Gamma)$ such that, if $\xi \in \cls(\Gamma)$, both nominals $a$ and $b$ occur in $\Gamma$, and $\xi'$ is obtained from $\xi$ by replacing every occurrence of $b$ with $a$, including the ones in automata of modal operators, then $\xi' \in \cls(\Gamma)$.
\myend
\end{Definition}

\begin{lemma}\label{lemma: HDKXS}
Let $n$ be the size of~$\Gamma$. Then, $|\clsZ(\Gamma)| = O(n^4)$ and $|\cls(\Gamma)| = O(2^{f(n)})$ for some polynomial $f(\cdot)$.
\end{lemma}

\begin{proof}
The cardinality of $\bsf(\Gamma)$ is of rank $O(n)$. Consider the construction of $\clsZ(\Gamma)$ by starting from $\bsf(\Gamma)$. Applying the first and second rules, we add $O(n^3)$ formulas to $\clsZ(\Gamma)$ (note that the number of states of an automaton $\Aut_\alpha$ is linear in the length of $\alpha$). After that, applying the third rule, we add $O(n^4)$ assertions to $\clsZ(\Gamma)$. Therefore, $|\clsZ(\Gamma)| = O(n^4)$. The second assertion of the lemma clearly follows. 
\myend
\end{proof}

Let $u$ be a complex node of $G$. For a formula/assertion $\xi$, by $\Repl(u)(\xi)$ we denote the formula/assertion obtained from $\xi$ by replacing every nominal $a$ with $\Repl(u)(a)$, including the ones in automata of modal operators. For a set $X$ of formulas/assertions, we define $\Repl(u)(X) = \{\Repl(u)(\xi) \mid \xi \in X\}$. 

\begin{lemma}\label{lemma: JKYEP}
Formulas and assertions used for the construction of any \CHPDL-tableau for $\Gamma$ belong to $\cls(\Gamma)$. 
Furthermore, for every node $v$ of $G$, the cardinality of $\FullLabel(v)$ is polynomial in the size of~$\Gamma$. 
\end{lemma}

\begin{proof}
The first assertion is clear. The second one follows from Lemma~\ref{lemma: HDKXS} and the observations:
if $u$ is a complex node of $G$, then $\Label(u) \cup \RFormulas(u) \subseteq \Repl(u)(\clsZ(\Gamma))$; 
if $u$ is a complex state of $G$ and $v$ is a descendant of $u$, then $\FullLabel(v) \subseteq \Repl(u)(\clsZ(\Gamma))$. 
\myend
\end{proof}

\begin{corollary}\label{cor: JDKSA}
Let $n$ be the size of~$\Gamma$. Then, the number of nodes of $G$ is (at most) exponential in~$n$. Consequently, $G$ can be constructed in exponential time in~$n$.
\end{corollary}

\begin{proof}
The first assertion holds because $G$ is constructed using global caching and due the facts stated by Lemma~\ref{lemma: JKYEP} and the second assertion of Lemma~\ref{lemma: HDKXS}. 
For the second assertion, just observe that each complex non-state may be re-expanded at most once. 
\myend
\end{proof}

	
	\subsection{Soundness}
	
	Our proof of soundness of the \CHPDL-tableau system relies on the notion of marking defined below.
	
	\begin{Definition}\label{def: marking}
		Let $\mM$ be a finitely branching Kripke model of $\Gamma$ and let $u$ be a complex state of $G$ such that $\Status(u) \neq \Incomplete$ and $\mM \models \FullLabel(u)$. Let $\noms' = \{a \in \noms \mid \Repl(u)(a) = a\}$ and $V' = \{v \in V \mid \SType(v) = \Simple\}$. A~{\em marking} of $G$ w.r.t.\ $\mM$ and $u$ is a function $f : O' \cup V' \to P(\Delta^\mM)$ with the intention that, for $x \in O' \cup V'$, $f(x)$ is the set of states of $\mM$ that ``correspond'' to $x$. It is defined to be the limit resulted from the following construction:
		\begin{itemize}
			\item for each $a \in \noms'$, set $f(a) := \{a^\mM\}$;
			\item for each $v \in V'$, set $f(v) := \emptyset$;
			\item initialize $U$ to a queue containing all the pairs $(a,a^\mM)$ for $a \in \noms'$ (in any order);
			\item while $U \neq \emptyset$, do:
			\begin{itemize}
				\item extract a pair $(x,y)$ from $U$;
				\item if $x \in \noms'$, then:
				\begin{itemize}
					\item for every $v \in V'$ such that $(u,v) \in E$, every $x\!:\!\lDmd{\sigma}\varphi \in \ELabels(u,v)$, and every $z \in (\Label(v))^\mM$ such that $(y,z) \in \sigma^\mM$,	 add $z$ to $f(v)$ and $(v,z)$ to~$U$;
				\end{itemize} 		    
				\item else if $\Type(x) = \State$, then:
				\begin{itemize}
					\item for every $v \in V'$ such that $(x,v) \in E$, every $\lDmd{\sigma}\varphi \in \ELabels(x,v)$, and every $z \in (\Label(v))^\mM$ such that $(y,z) \in \sigma^\mM$, add $z$ to $f(v)$ and $(v,z)$ to~$U$;
				\end{itemize} 
				\item else if there exists $a \in \Label(x)$, then:
				\begin{itemize}
					\item add $y$ to $f(a)$ and $(a,y)$ to~$U$;
				\end{itemize}
				\item else:
				\begin{itemize}
					\item for every $v \in V'$ such that $(x,v) \in E$ and $y \in (\Label(v))^\mM$, add $y$ to $f(v)$ and $(v,y)$ to~$U$.
					\myend
				\end{itemize} 
			\end{itemize}
		\end{itemize}	
	\end{Definition}
	
	\begin{lemma}\label{lemma: IUSNA}
	Every path consisting of only non-states in $G$ is finite.
	\end{lemma}
	
	\begin{proof}
	This lemma follows from the following observations:
	\begin{itemize}
		\item If a non-state $w$ is a successor of a non-state $v$ then $\RFormulas(w) \supset \RFormulas(v)$ or $\FullLabel(w) \supset \FullLabel(v)$ or the number of nominals occurring in $\Label(w)$ is smaller than the number of nominals occurring in $\Label(v)$. 
		\item For any node $w$, $\RFormulas(w)$ and $\FullLabel(w)$ are subsets of the finite set $\cls(\Gamma)$.
		\myend
	\end{itemize}
	\end{proof}
	
	\begin{lemma}\label{lemma: IDJSO}
		Let $\mM$ be a finitely branching Kripke model of $\Gamma$ and $u$ a complex state of $G$ such that $\Status(u) \neq \Incomplete$ and $\mM \models \FullLabel(u)$. Then $\Status(u) \neq \Unsat$.
	\end{lemma}
	
	\begin{proof}
		Let $f : O' \cup V' \to P(\Delta^\mM)$ a marking of $G$ w.r.t.\ $\mM$ and $u$, where $\noms' = \{a \in \noms \mid \Repl(u)(a) = a\}$ and $V' = \{v \in V \mid \SType(v) = \Simple\}$. 		
		Let $V'' = \{u\} \cup \{w \in V' \mid f(w) \neq \emptyset\}$. 
		We prove that, if the status of a node $v \in V$ is changed to $\Unsat$ or $\UnsatWrtP{u}$, then $v \notin V''$, by induction on that moment. 
		
		Recall that $\mM \models \FullLabel(u)$ and observe that, if $w \in V'$ and $z \in f(w)$, then $z \in (\Label(w))^\mM$ and thus $z \in (\FullLabel(w))^\mM$. Hence, if $w \in V''$, then $\FullLabel(w)$ is satisfiable. 
		
		If $\Status(v)$ is changed to $\Unsat$ by the rule $\rUnsat$ because there exists $o\!:\!\bot \in \Label(v)$ or $a\!:\!\lnot a \in \Label(v)$ or $\{\xi,\ovl{\xi}\} \subseteq \FullLabel(v)$, then $\FullLabel(v)$ is unsatisfiable and hence $v \notin V''$. 
		
		If $v \in V''$ and $\Status(v)$ was changed to $\Unsat$ by the rule $\rUnsat$ because $\Status(v) = \UnsatWrtP{v}$, then $v$ must be a complex state and thus $v = u$, and by the inductive assumption, $v \notin V''$, a contradiction. 
		
		Consider the case when $Status(v)$ is changed to $\UnsatWrtP{u}$ by the rule $\rNom$ and, for the sake of contradiction, assume that $v \in V''$. Thus, $v \in V'$ and $f(v) \neq \emptyset$. Hence, $\FullLabel(v)$ is satisfied at a state in $\mM$, in particular, $\mM \models \xi$ (where $\xi$ is the assertion mentioned in the rule $\rNom$), which contradicts the facts that $\ovl{\xi} \in \FullLabel(u)$ and $\mM \models \FullLabel(u)$. 
		
		Observe that, if $w \in V''$, $\Status(w) \notin \{\Unexpanded$, $\Unsat$, $\Blocked$, $\UnsatWrt(\ldots)\}$ and $\Type(w) = \NonState$, then $w$ must have a successor belonging to $V''$. Hence, if $\Status(v)$ is changed to $\Unsat$ or $\UnsatWrtP{u}$ by the instruction~2 of the rule $\rUnsat$, then, by the inductive assumption, it follows that $v \notin V''$. 
		
		Similarly, if $w \in V''$, $\Status(w) \notin \{\Unexpanded,\Unsat,\Incomplete\}$ and $\Type(w) = \State$, then all successors of $w$ must belong to $V''$. Hence, if $\Status(v)$ is changed to $\Unsat$ or $\UnsatWrtP{u}$ by the instruction~3 of the rule $\rUnsat$, then, by the inductive assumption, it follows that $v \notin V''$. 
		
		There remains the case when $\Status(v)$ is changed to $\UnsatWrtP{u}$ by the instruction 1b of the rule $\rUnsat$. For this case, it is sufficient to prove that:
		\begin{enumerate}
			\item every assertion of the form $a\!:\!\lDmd{A,q}\varphi$ or $a\!:\!\lDmd{\omega}\lDmd{A,q}\varphi$ in $\FullLabel(u)$ is $\Dmd$-realizable at $u$ w.r.t.~$u$,
			\item if $w \in V'$ and $f(w) \neq \emptyset$, then every formula of the form $\lDmd{A,q}\varphi$ or $\lDmd{\omega}\lDmd{A,q}\varphi$ in $\FullLabel(w)$ is $\Dmd$-realizable at $w$ w.r.t.~$u$.
		\end{enumerate}
		
		Consider the second assertion and suppose that $w \in V'$, $z \in f(w)$, $\xi \in \FullLabel(w)$ and $\xi$ is of the form $\lDmd{A,q}\varphi$ or $\lDmd{\omega}\lDmd{A,q}\varphi$. We have that $z \in (\Label(w))^\mM$. Consider the case when $\xi = \lDmd{A,q}\varphi$ (the other case is similar and omitted). Since $z \in (\Label(w))^\mM$, there exist $z_0,\ldots,z_k \in \Delta^\mM$ and a~word $\omega_1\ldots\omega_k$ accepted by $(A,q)$ such that $z_0 = z$, $z_k \in \varphi^\mM$ and, for each $1 \leq i \leq k$, if $\omega_i \in \mindices$ then $(z_{i-1},z_i) \in \omega_i^\mM$, else $\omega_i = (\psi_i?)$ for some $\psi_i$ and $z_{i-1} = z_i$ and $z_i \in \psi_i^\mM$. Such a realization (satisfaction) of $\xi$ at $z$ in $\mM$ is reflected by a sequence $(w_0,\xi_0), \ldots, (w_h, \xi_h)$ such that $w_0,\ldots,w_h \in V''$, $w_0 = w$, $\xi_0 = \xi$, $\xi_i$ is of the form $o_i\!:\!\lDmd{A,q_i}\varphi$ or $o_i\!:\!\lDmd{\omega'_i}\lDmd{A,q_i}\varphi$ (for $0 \leq i \leq h$), $\xi_h$ is $\Dmd$-realizable at $w_h$ w.r.t.~$u$ by the condition~\ref{item: YUDSS a} or~\ref{item: YUDSS e} of Definition~\ref{def: YUDSS}, and each $\xi_i$ with $0 \leq i < h$ is $\Dmd$-realizable at $w_i$ w.r.t.~$u$ because $\xi_{i+1}$ is $\Dmd$-realizable at $w_{i+1}$ w.r.t.~$u$, due to a condition among \ref{item: YUDSS b}, \ref{item: YUDSS c}, \ref{item: YUDSS f}--\ref{item: YUDSS k} of Definition~\ref{def: YUDSS}. For this claim, we use the inductive assumption and the fact that every path consisting of only non-states in $G$ is finite (Lemma~\ref{lemma: IUSNA}). As a consequence, $\xi$ is $\Dmd$-realizable at $w$ w.r.t.~$u$. 
		
		The first assertion in the above list can be proved analogously. 
		\myend
	\end{proof}
	
	\begin{corollary}[Soundness]\label{cor: Soundness}
		Let $\Gamma$ be an ABox in NNF and $G = (V,E,\nu)$ an arbitrary \CHPDL-tableau for~$\Gamma$. If $\Gamma$ is satisfiable, then $\Status(\nu) \neq \Unsat$.
	\end{corollary}
	
	\begin{proof}
		Assume that $\Gamma$ is satisfiable. It is well known that, if an ABox (in \HPDL) is satisfiable, then it has a finitely-branching Kripke model. Let $\mM$ be a finitely-branching Kripke model of $\Gamma$. Since $\Label(\nu) = \Gamma$ and $\mM \models \Gamma$, there exists a complex state $u$ of $G$ such that $\Status(u) \neq \Incomplete$ and $\mM \models \FullLabel(u)$. By Lemma~\ref{lemma: IDJSO}, $\Status(u) \neq \Unsat$. This implies that $\Status(\nu) \neq \Unsat$.
		\myend
	\end{proof}
	

	\subsection{Completeness}
	
	In this subsection, assume that $\Status(\nu) \neq \Unsat$ (where $\nu$ is the root of $G$). We prove that $\Gamma$ is satisfiable by constructing a {\em model graph} $G'$ of $G$ and a Kripke model $\mM$ that corresponds to~$G'$. 
	
	Since $\Status(\nu) \neq \Unsat$, there exists a complex state $u$ of $G$ such that $\Status(u) \notin \{\Unsat,\Incomplete\}$. In this subsection, we fix such a $u$. We also fix a sequential process of marking assertions/formulas of the form $o\!:\!\lDmd{A,q}\varphi$ or $o\!:\!\lDmd{\omega}\lDmd{A,q}\varphi$ in the full labels of nodes of $G$ as $\Dmd$-realizable w.r.t.~$u$. For $v \in V$ and $\xi \in \FullLabel(v)$ of one of these forms, let $\TSDR(\xi,v, u)$ denote the moment at which $\xi$ is marked as $\Dmd$-realizable at~$v$ w.r.t.~$u$. In the other case, $\TSDR(\xi,v, u)$ is undefined.  
	
	\begin{Definition}
		Let $v$ be a descendant of $u$ with $\Status(v) \notin \{\Unsat$, $\UnsatWrtP{u}\}$ and let $\xi \in \FullLabel(v)$ be of the form \mbox{$o\!:\!\lDmd{A,q}\varphi$} or \mbox{$o\!:\!\lDmd{\omega}\lDmd{A,q}\varphi$}. A {\em $\Dmd$-realization of $\xi$ at $v$ w.r.t.~$u$} is a sequence $(v_0,\xi_0), \ldots, (v_{k+1},\xi_{k+1})$ such that:
		\begin{itemize}
			\item $k \geq 0$, $v_0 = v$ and $\xi_0 = \xi$,
			\item for each $0 \leq i \leq k$, $\xi_i \in \FullLabel(v_i)$, $\xi_i$ is of the form $o_i\!:\!\lDmd{A,q_i}\varphi$ or $o_i\!:\!\lDmd{\omega_i}\lDmd{A,q_i}\varphi$ and is $\Dmd$-realizable at $v_i$ w.r.t.~$u$, 
			\item for each $0 \leq i < k$, $\TSDR(\xi_i,v_i,u) > \TSDR(\xi_{i+1},v_{i+1},u)$ and $\xi_i$ was marked at the moment $\TSDR(\xi_i,v_i,u)$ as $\Dmd$-realizable at $v_i$ w.r.t.~$u$ due to the $\Dmd$-realizability of $\xi_{i+1}$ at $v_{i+1}$ w.r.t.~$u$ according to the rules specified by the conditions \ref{item: YUDSS b}, \ref{item: YUDSS c}, \ref{item: YUDSS f}--\ref{item: YUDSS k} of Definition~\ref{def: YUDSS},
			\item one of the following two conditions holds:
			\begin{itemize}
				\item $\xi_k$ was marked at the moment $\TSDR(\xi_k,v_k,u)$ as $\Dmd$-realizable at $v_k$ w.r.t.~$u$ due to the rule specified by the condition~\ref{item: YUDSS e} of Definition~\ref{def: YUDSS}, $v_k$ was expanded by the tableau rule $\rDmdD$ with \mbox{$\xi_k = o_k\!:\!\lDmd{A,q_k}\varphi$} being the principal formula (having $q_k \in F_A$), $v_{k+1}$ is the successor of $v_k$ whose label is obtained from $\Label(v_k)$ by replacing $\xi_k$ with $\xi_{k+1} = o_{k+1}\!:\!\varphi$, and $\Status(v_{k+1}) \notin \{\Unsat,\UnsatWrtP{u}\}$; 
				\item $\xi_k$ was marked at the moment $\TSDR(\xi_k,v_k,u)$ as $\Dmd$-realizable at $v_k$ w.r.t.~$u$ due to the rule specified by the condition~\ref{item: YUDSS a} of Definition~\ref{def: YUDSS}, \mbox{$\xi_k = o_k\!:\!\lDmd{A,q_k}\varphi$}, $q_k \in F_A$, $o_{k+1}\!:\!\varphi \in \FullLabel(v_k)$, $v_{k+1} = v_k$ and $\xi_{k+1} = o_{k+1}\!:\!\varphi$. 
				\myEnd
			\end{itemize}
		\end{itemize}
	\end{Definition}
	
	\begin{lemma}\label{lemma: IUFJS}
	If $v$ is a descendant of $u$ with $\Status(v) \notin \{\Unsat$, $\UnsatWrtP{u}\}$, then every $\xi \in \FullLabel(v)$ of the form \mbox{$o\!:\!\lDmd{A,q}\varphi$} or \mbox{$o\!:\!\lDmd{\omega}\lDmd{A,q}\varphi$} has a $\Dmd$-realization at~$v$ w.r.t.~$u$.
	\end{lemma}
	
	This lemma clearly holds, because $\xi$ has a $\Dmd$-realization at~$v$ w.r.t.~$u$ iff it is $\Dmd$-realizable at~$v$ w.r.t.~$u$.

	\begin{Definition}
		Let $v$ be a descendant of $u$ such that it is a simple non-state and $\Status(v) \notin \{\Unsat$, $\UnsatWrtP{u}\}$. A {\em saturation path} of $v$ w.r.t.~$u$ is a sequence $v_0$, $v_1$, \ldots, $v_k$ of nodes of $G$, with $v_0 = v$ and $k \geq 1$, such that:
		\begin{itemize}
			\item $\Status(v_i) \notin \{\Unsat$, $\UnsatWrtP{u}\}$ for all $0 \leq i \leq k$, 
			\item $\Type(v_i) = \NonState$ for all $0 \leq i < k$ and $\Type(v_k) = \State$, 
			\item $\tuple{v_i,v_{i+1}} \in E$ for all $0 \leq i < k-1$, 
			\item if there exists $a \in \Label(v_{k-1})$, then $v_k = u$, else $\tuple{v_{k-1},v_k} \in E$.
			\myEnd
		\end{itemize}
	\end{Definition}
	
	By Lemma~\ref{lemma: IUSNA}, each saturation path of $v$ w.r.t.~$u$ is finite. Furthermore, if $v_i$ is a simple non-state with $\Status(v_i) \notin \{\Unsat$, $\UnsatWrtP{u}\}$ and $\Label(v_i)$ does not contain any nominal, then $v_i$ has a successor $v_{i+1}$ with $\Status(v_{i+1}) \notin \{\Unsat$, $\UnsatWrtP{u}\}$. Therefore, we have the following lemma.
	
	\begin{lemma}\label{lemma: JDLAB}
		If $v$ is a descendant of $u$ such that it is a simple non-state and $\Status(v) \notin \{\Unsat$, $\UnsatWrtP{u}\}$, then it has at least one saturation path w.r.t.~$u$.
	\end{lemma}
	
	\begin{Definition}
	A {\em model graph} of $G$ w.r.t.~$u$ is a structure $G' = (V',E',\Label',\ELabels')$ constructed as follows, where $V' \subseteq \noms \cup \{v \in V \mid \SType(v) = \Simple$ and $\Type(v) = \State\}$, $E' \subseteq V' \times V'$, $\ELabels': E' \to P(\mindices)$, and for $x \in V'$, $\Label'(x)$ is a set of formulas:
	\begin{enumerate}
		\item $V' := \{a \in \noms \mid \Repl(u)(a) = a\}$;
		\item $E' := \{(a,b) \mid \textrm{there exists $\sigma(a,b) \in \Label(u)$}\}$;
		\item for each $a \in V'$, $\Label'(a) := \{\varphi \mid a\!:\!\varphi \in \FullLabel(u)\} \cup \{a\}$;
		\item for each $(a,b) \in E'$, $\ELabels(a,b) := \{\sigma \mid \sigma(a,b) \in \Label(u)\}$;
		\item $\realized := \emptyset$;
		\item while there exist $x \in V'$ and $\varphi = \lDmd{\sigma}\lDmd{A,q}\psi \in \Label'(x)$ such that $(x,\varphi) \notin \realized$, do:
		\begin{enumerate}
			\item\label{step: DSSFS 6a} if $x \in V$, then let $(v_0,\xi_0), \ldots, (v_{k+1},\xi_{k+1})$ be a $\Dmd$-realization of $\varphi$ at $x$ w.r.t.~$u$, else let $(v_0,\xi_0), \ldots, (v_{k+1},\xi_{k+1})$ be a $\Dmd$-realization of $x\!:\!\varphi$ at $u$ w.r.t.~$u$;
			\item\label{step: DSSFS 6b} if $\Type(v_{k+1}) = \State$, then let $l = 1$, else let $v_{k+1},\ldots,v_{k+l}$ be a saturation path of $v_{k+1}$ w.r.t.~$u$;
			\item\label{step: DSSFS 6c} let $i_1, \ldots, i_h$ be all the indices such that 
				$0 < i_1 < \ldots < i_h = k+l$,
				$\Type(v_{i_j}) = \State$ for $1 \leq j \leq h$,
				and for every $1 \leq j < h$, if $v_{i_j} = u$, then $\xi_{i_j}$ is of the form $a_j\!:\!\lDmd{\sigma_j}\lDmd{A,q_{i_j}}\psi$, else $\xi_{i_j}$ is of the form $\lDmd{\sigma_j}\lDmd{A,q_{i_j}}\psi$;
			\item if $v_{i_h} = u$, then let $a_h$ be a nominal such that $a_h \in \Label(v_i)$, where $i$ is the greatest index such that $i < i_h$ and $v_i \neq u$;
			\item $x_0 := x$, $\sigma_0 := \sigma$;
			\item for each $j$ from 1 to $h$, do:
			\begin{enumerate}
				\item if $v_{i_j} \neq u$, then: 
				\begin{itemize}
					\item $x_j := v_{i_j}$;
					\item if $x_j \notin V'$, then add $x_j$ to $V'$ and set $\Label'(x_j) := \FullLabel(v_{i_j-1})$;
					\item else $\Label'(x_j) := \Label'(x_j) \cup \RFormulas(v_{i_j-1})$;
					\item add $(x_{j-1},x_j)$ to $E'$ and $\sigma_{j-1}$ to $\ELabels(x_{j-1},x_j)$;
				\end{itemize}
				\item else: $x_j := a_j$,  
						add $(x_{j-1},x_j)$ to $E'$ and $\sigma_{j-1}$ to $\ELabels(x_{j-1},x_j)$;
			\end{enumerate}
			\item add $(x,\varphi)$ to $\realized$.
			\myend
		\end{enumerate}
	\end{enumerate}		
	\end{Definition}
	As invariants of the ``while'' loop in the above construction, we have that: 
	\begin{itemize}
		\item if $x \in V'$, then $\Status(x) \notin \{\Unsat$, $\UnsatWrtP{u}\}$;
		\item the ``let'' instruction at the step~\ref{step: DSSFS 6a} is well-defined due to Lemma~\ref{lemma: IUFJS}; 
		\item the ``let'' instruction at the step~\ref{step: DSSFS 6b} is well-defined due to Lemma~\ref{lemma: JDLAB}. 
	\end{itemize} 
	Also observe that:
	\begin{itemize}
		\item the ``let'' instruction at the step~\ref{step: DSSFS 6c} is well-defined; it specifies not only $i_1, \ldots, i_h$, but also $\sigma_j$ for all $1 \leq j < h$, and $a_j$ for all $1 \leq j < h$ such that $v_{i_j} = u$;
		\item since $G$ is finite (by Corollary~\ref{cor: JDKSA}), the ``while'' loop terminates and $G'$ is finite.
	\end{itemize} 
	
	The following lemma states that a model graph is similar to a Hintikka structure. It directly follows from the constructions of~$G$ and~$G'$.
	
	\begin{lemma}\label{lemma: HDIAK}
		Let $G' = (V',E',\Label',\ELabels')$ be a model graph of $G$ w.r.t.~$u$. Then, for every $x \in V'$ and every $\varphi \in \Label'(x)$:
		\begin{enumerate}
			\item $\bot \notin \Label'(x)$ and $\ovl{\varphi} \notin \Label'(x)$, 
			\item if $\varphi = \psi \land \chi$, then $\{\psi,\chi\} \subset \Label'(x)$,
			\item if $\varphi = \psi \lor \chi$, then $\psi \in \Label'(x)$ or $\chi \in \Label'(x)$,
			\item if $\varphi = a$, then $x = a$, 
			\item if $\varphi = [\alpha]\psi$, $\alpha \notin \mindices$, $\alpha$ is not a~test and $I_{\Aut_\alpha} = \{q_1,\ldots,q_k\}$,\\ then $\{[\Aut_\sigma,q_1]\psi,\ldots,[\Aut_\sigma,q_k]\psi\} \subset \Label'(x)$,
			\item if $\varphi = [A,q]\psi$ and $\delta_A(q) = \{(\omega_1,q_1)$, \ldots, $(\omega_k,q_k)\}$,\\ then $\{[\omega_1][A,q_1]\psi$, \ldots, $[\omega_k][A,q_k]\psi\} \subset \Label'(x)$,
			\item if $\varphi = [A,q]\psi$ and $q \in F_A$, then $\psi \in \Label'(x)$,
			\item if $\varphi = [\chi?]\psi$, then $\ovl{\chi} \in \Label'(x)$ or $\psi \in \Label'(x)$,
			\item if $\varphi = [\sigma]\psi$, $(x,y) \in E'$ and $\sigma \in \ELabels'(x,y)$, then $\psi \in \Label'(y)$,
			\item if $\varphi = \lDmd{\alpha}\psi$, $\alpha \notin \mindices$, $\alpha$ is not a~test and $I_{\Aut_\alpha} = \{q_1,\ldots,q_k\}$,\\ then $\{\lDmd{\Aut_\alpha,q_1}\psi, \ldots, \lDmd{\Aut_\alpha,q_k}\psi\} \cap \Label'(x) \neq \emptyset$, 
			\item if $\varphi = \lDmd{\chi?}\psi$, then $\{\chi,\psi\} \subseteq \Label'(x)$,
			\item if $\varphi = \lDmd{\sigma}\psi$, then there exists $y$ such that $(x,y) \in E'$, $\sigma \in \ELabels'(x,y)$ and $\psi \in \Label'(y)$,
			\item if $\varphi = \lDmd{A,q}\psi$, then there exist an accepting run $q_0,\ldots,q_k$ of the automaton $(A,q)$ on a word $\omega_1\ldots\omega_k$ (with $q_0 = q$ and $q_k \in F_A$) and a sequence $x_0,\ldots,x_k$ of nodes of $G'$ such that $x_0 = x$, $\psi \in \Label'(x_k)$ and, for each $1 \leq i \leq k$, if $\omega_i = (\chi_i?)$, then $x_i = x_{i-1}$ and $\{\chi_i,\lDmd{A,q_i}\psi\} \subseteq \Label'(x_i)$, else $(x_{i-1},x_i) \in E'$, $\omega_i \in \ELabels'(x_{i-1},x_i)$ and $\lDmd{A,q_i}\psi \in \Label'(x_i)$.
		\end{enumerate}
	\end{lemma}
	
	\begin{Definition}
		Let $G' = (V',E',\Label',\ELabels')$ be a model graph of $G$ w.r.t.~$u$. A Kripke model $\mM$ {\em corresponds} to $G'$ w.r.t.~$u$ if:
		\begin{itemize}
			\item $\Delta^\mM = V'$,
			\item $p^\mM = \{x \in V' \mid p \in \Label'(x)\}$ for $p \in \props$, 
			\item $\sigma^\mM = \{(x,y) \in E' \mid \sigma \in \ELabels'(x,y)\}$ for $\sigma \in \mindices$, 
			\item $a^\mM = \Repl(u)(a)$ for $a \in \noms$ with $\Repl(u)(a)$ specified.
			\myend
		\end{itemize}
	\end{Definition}
	
	Clearly, there exist Kripke models corresponding to $G'$ w.r.t.~$u$. They differ from each other only in interpreting nominals $a$ with $\Repl(u)(a)$ unspecified.
	
	\begin{lemma}\label{lemma: HIWJC}
	Let $G' = (V',E',\Label',\ELabels')$ be a model graph of $G$ w.r.t.~$u$ and $\mM$ a Kripke model corresponding to $G'$ w.r.t.~$u$. Then:
	\begin{enumerate}
		\item for every $x \in V'$ and every $\varphi \in \Label'(x)$, we have $x \in \varphi^\mM$,
		\item $\mM \models \FullLabel(u)$,
		\item $\mM \models \Label(\nu)$.
	\end{enumerate}
	\end{lemma}
	
	\begin{proof}
	The first assertion can be proved in a straightforward way by induction on the structure of $\varphi$ using Lemma~\ref{lemma: HDIAK}. The second assertion follows from the first one, the initialization of the construction of $G'$ and the interpretation of nominals in~$\mM$. The third assertion follows from the second one. Namely, there exists a path $v_0,\ldots,v_k$ in $G$ such that $v_0 = \nu$ and $v_k = u$, and by the applied tableau rules, for every $i$ from $k-1$ down to 0, $\mM \models \FullLabel(v_i)$ follows from $\mM \models \FullLabel(v_{i+1})$.
	\myend
	\end{proof}
		
	\begin{corollary}[Completeness]\label{cor: Completeness}
		Let $\Gamma$ be an ABox in NNF and $G = (V,E,\nu)$ an arbitrary \CHPDL-tableau for~$\Gamma$. If $\Status(\nu) \neq \Unsat$, then $\Gamma$ is satisfiable.
	\end{corollary}
	
	This corollary follows from the third assertion of Lemma~\ref{lemma: HIWJC} (since $\Label(\nu) = \Gamma$).
	

\section{Concluding Remarks}\label{section: conc}

We have given the first direct tableau procedure with the \EXPTIME complexity for deciding HPDL and proved that it is sound and complete. The procedure uses global caching, a technique that not only guarantees the \EXPTIME complexity, but also increases efficiency. 
As HPDL can be used as a description logic for representing and reasoning about terminological knowledge, our procedure is useful for practical applications. 

In our decision procedure, any expansion strategy can be used for constructing a tableau. One can give the instruction 1(a) of the rule $\rUnsat$ the highest priority and give unary static rules a higher priority than for non-unary static rules. One may choose the depth-first expansion strategy, globally cache only simple nodes, keep complex nodes only for the current path of complex nodes, and is still guaranteed to have the \EXPTIME complexity for the algorithm. Checking fulfillment of eventualities can be done on-the-fly as in~\cite{AbateGW09} or periodically for the whole graph or at special moments for subgraphs (for example, when the subgraph rooted at a node has been ``fully expanded'' and no $\Dmd$-realization goes out from that subgraph). 

Our decision procedure has been designed to simplify the presentation and leaves space for improvement. 
It has been implemented for the TGC2 system~\cite{TGC2} with various optimizations. For example, a sequence of expansions by unary static rules is done in one step to eliminate intermediate non-states with only one successor, automata in modal operators are minimized, and different control strategies (for expanding the constructed tableau) are mixed. As TGC2 aims to allow efficient automated reasoning in a large class of modal and description logics, its implementation is time-consuming. A few intended important optimization techniques like propagation of unsatisfiability cores or compacting ABoxes by using bisimilarity were not implemented for TGC2 yet. In general, TGC2 still needs improvements, at least with respect to functionality. We refer the reader to~\cite{TGC2} for more details about this system. 

Our tableau method can be extended for Graded HPDL using the techniques from~\cite{GPDL} and for Converse-HPDL using the techniques from~\cite{GoreW10,nCPDLreg-AMSTA,SHIO}.




\end{document}